\documentclass[11pt]{article}

\usepackage{tabu}
\usepackage{array}
\usepackage{amsmath}
\usepackage{amsthm}
\usepackage{amssymb}
\usepackage{mathtools}
\usepackage{thm-restate}
\usepackage{authblk}
\usepackage{fullpage}
\usepackage{mathtools}
\usepackage{bbm}
\usepackage{textcomp}
\usepackage{multirow}
\usepackage{comment}
\usepackage[T1]{fontenc}
\usepackage[linesnumbered,ruled,procnumbered]{algorithm2e}

\newtheorem{theorem}{Theorem}[section]
\newtheorem{lemma}[theorem]{Lemma}
\newtheorem{definition}[theorem]{Definition}

\newtheorem{corollary}[theorem]{Corollary}

\newtheorem{remark}[theorem]{Remark}
\newtheorem{claim}[theorem]{Claim}

\newtheorem{hypothesis}[theorem]{Hypothesis}
\newtheorem{observation}[theorem]{Observation}
\usepackage{tikz}
\usepackage{hyperref}  
\usepackage[capitalize]{cleveref}
\hypersetup{colorlinks=true,citecolor=red,linkcolor=blue}
\hypersetup{
    colorlinks=true,
    citecolor=blue,
    linkcolor=blue,
    filecolor=cyan,
    urlcolor=magenta,
}
\usetikzlibrary{arrows}

\newcommand{\wt}{\widetilde}

\newcommand{\eps}{\epsilon}
\newcommand{\N}{\mathbb{N}}
\newcommand{\A}{\mathcal{A}}

\newcommand{\R}{\mathbb{R}}

\renewcommand{\varepsilon}{\lambda}
\renewcommand{\tilde}{\wt}

\renewcommand{\eps}{\epsilon}

\newcommand{\ip}[2]{\langle {#1} , {#2} \rangle}
\newcommand{\norm}[1]{\| #1 \|}

\DeclareMathOperator*{\Z}{\mathbb{Z}}

\newcommand{\ket}[1]{\left| #1 \right\rangle}

\DeclareMathOperator{\poly}{poly}
\DeclareMathOperator{\polylog}{polylog}

\newcommand{\ceil}[1]{\lceil #1 \rceil}
\newcommand{\floor}[1]{\lfloor #1 \rfloor}

\newcommand{\id}{\mathrm{id}}

\newcommand{\ora}{\mathcal{O}}
\newcommand{\FG}{\mathrm{FG}}
\newcommand{\QFG}{\mathrm{QFG}}
\newcommand{\CP}{\ensuremath{\mathsf{CP}}}
\newcommand{\OV}{\ensuremath{\mathsf{OV}}}
\newcommand{\OVH}{\ensuremath{\mathsf{OVH}}}
\newcommand{\SAT}{\ensuremath{\mathsf{SAT}}}
\newcommand{\BCP}{\ensuremath{\mathsf{BCP}}}
\newcommand{\xiBCP}{(1+\xi)\text{-}\mathsf{BCP}}
\newcommand{\xiCP}{(1+\xi)\text{-}\mathsf{CP}}
\newcommand{\eCP}{\mathsf{CP}_{\epsilon}}
\newcommand{\eBCP}{\mathsf{BCP}_{\epsilon}}
\newcommand{\xieBCP}{(1+\xi)\text{-}\mathsf{BCP}_{\epsilon}}
\newcommand{\ZOV}{\ensuremath{\Z\text{-}\mathsf{OV}}}

\crefname{claim}{Claim}{Claims}
\crefname{procedure}{Procedure}{Procedures}

\SetAlgoProcName{Procedure}{Procedure}

\definecolor{mygreen}{RGB}{80,180,0}
\definecolor{b2}{RGB}{51,153,255}

\newcommand{\nc}{\newcommand}

\nc{\nnl}{\nn \\ &}  
\nc{\fot}{\frac{1}{2}} 
\nc{\oo}[1]{\frac{1}{#1}} 
\newcommand{\ben}{\begin{enumerate}}
\newcommand{\een}{\end{enumerate}}
\nc{\mc}{\mathcal}

\nc{\onenorm}[1]{\L\| #1 \R\|_1} 

\nc{\Ra}{\Rightarrow}
\nc{\zo}{\{0,1\}}	

\title{On the Quantum Complexity of Closest Pair and Related Problems}

\author{Scott Aaronson\thanks{ Supported by a Vannevar Bush Fellowship from the US Department of Defense, a Simons Investigator Award, and the Simons ``It from Qubit'' collaboration.}}
\author{Nai-Hui Chia\thanks{ Supported by Aaronson's Vannevar Bush Faculty Fellowship from the US Department of Defense.}}
\author{Han-Hsuan Lin$^{\dag}$}
\author{Chunhao Wang$^{\dag}$}
\author{Ruizhe Zhang\thanks{ Supported by NSF Grant CCF-1648712.}}
\affil{Department of Computer Science, University of Texas at Austin.  \\ Email: \{aaronson,nai,linhh,chunhao,rzzhang\}@cs.utexas.edu}
\date{\empty}

\begin{document}

\begin{titlepage}
\maketitle
\begin{abstract}
The closest pair problem is a fundamental problem of computational geometry: given a set of $n$ points in a $d$-dimensional space, find a pair with the smallest distance. A classical algorithm taught in introductory courses solves this problem in $O(n\log n)$ time in constant dimensions (i.e., when $d=O(1)$). This paper asks and answers the question of the problem's quantum {time} complexity. Specifically, we give an $\tilde{O}(n^{2/3})$ algorithm in constant dimensions, which is optimal up to a polylogarithmic factor by the lower bound on the quantum query complexity of element distinctness. The key to our algorithm is an efficient history-independent data structure that supports quantum interference.

In $\polylog(n)$ dimensions, no known quantum algorithms perform better than brute force search, with a quadratic speedup provided by Grover's algorithm. To give evidence that the quadratic speedup is nearly optimal, we initiate the study of quantum fine-grained complexity and introduce the \emph{Quantum Strong Exponential Time Hypothesis (QSETH)}, which is based on the assumption that Grover's algorithm is optimal for \textsf{CNF-SAT} when the clause width is large.  We show that the na\"{i}ve Grover approach to closest pair in higher dimensions is optimal up to an $n^{o(1)}$ factor unless QSETH is false. We also study the bichromatic closest pair problem and the orthogonal vectors problem, with broadly similar results. 
\end{abstract}
\thispagestyle{empty}
\end{titlepage}

\section{Introduction}

In the closest pair problem ($\CP$), we are given a list of points in $\mathbb{R}^d$, and asked to find two that are closest.  (See \cref{fig:cp} for an illustration of this problem.) This is a fundamental problem in computational geometry and has been extensively studied. Indeed, $\CP$ is one of the standard examples in textbooks (such as~\cite{CLRS09} and~\cite{KT06}) to introduce the divide-and-conquer technique. Moreover, $\CP$ relates to problems that have critical applications in spatial data analysis and machine learning, such as empirical risk minimization~\cite{backurs2017fine}, point location~\cite{SH75,bespamyatnikh98}, time series motif mining~\cite{mueen09}, spatial matching problems~\cite{wong07}, and clustering~\cite{nan01}. Therefore, any improvement on $\CP$ may imply new efficient algorithms for related applications. 

\begin{figure}[h]
  \centering
  \includegraphics[width=0.6\textwidth]{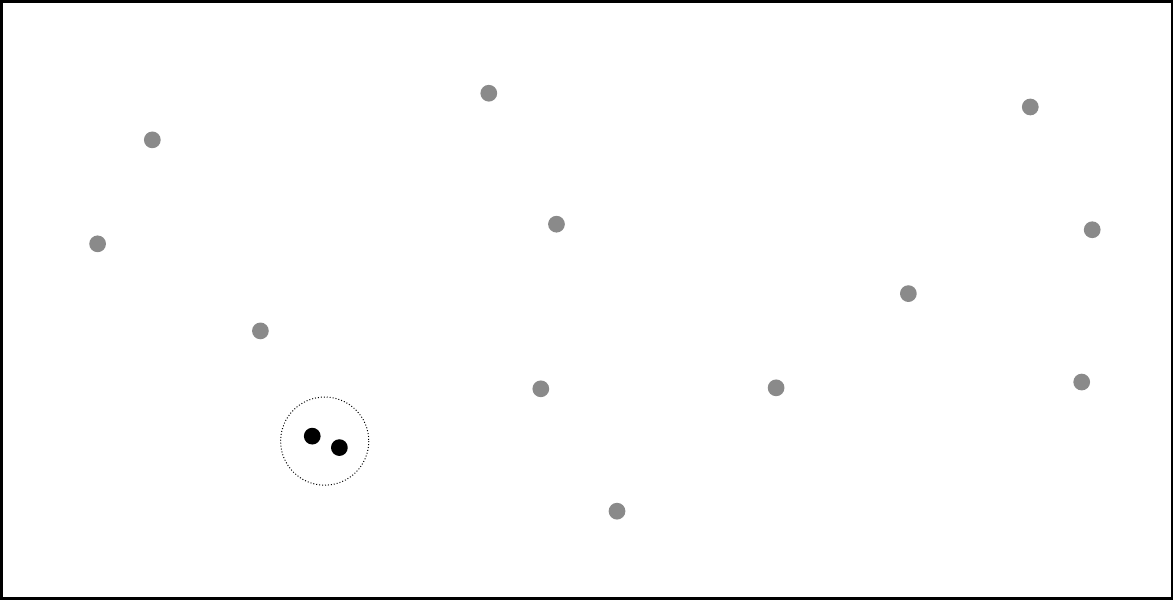}
  \caption{An instance of the $\CP$, where the the closest pair is labeled in the circle.\label{fig:cp}}
\end{figure}

Like with many other geometric problems, the hardness of $\CP$ rises as the dimension $d$ increases. Shamos and Hoey gave the first $O(n\log n)$ deterministic algorithm in $\mathbb{R}^2$ by using Voronoi diagrams~\cite{SH75}, improving on the trivial $O(n^2d)$ upper bound. Then, Bentley and Shamos gave an algorithm with  $2^{O(d)} n\log n$ running time via a divide-and-conquer approach~\cite{bs76}. A randomized algorithm by Khuller and Matias~\cite{KM95,rab76}
takes $2^{O(d)} n$ expected running time. A trivial lower bound for $\CP$ is $\Omega(n)$, since one must read all points to find the closest pair in the worst case. Yao showed an $\Omega(n \log n)$ lower bound for $\CP$ on the algebraic decision tree model~\cite{Yao89}. 

When we consider $\CP$ in $\polylog(n)$ dimensions, the running time of all existing algorithms blows up to $\Omega(n^2)$, and thus it is unknown if there exists an algorithm matching the unconditional lower bounds. Nevertheless, under the \emph{Strong Exponential Time Hypothesis (SETH)}, Karthik and Manurangsi~\cite{km19}, and David et al.~\cite{dkl18}, recently proved a conditional lower bound of $n^{2-o(1)}$ for $\CP$ in $\polylog(n)$ dimensions. This implies that the brute force approach is nearly optimal in $\polylog(n)$ dimensions unless SETH is false.  SETH was introduced by Impagliazzo and Paturi~\cite{IP01}, and is the assumption that for all $\epsilon>0$, there exists an integer $k>2$ such that no algorithm can solve $k$-SAT in time $O(2^{(1-\epsilon)n})$. 
 
The main idea behind the results of~\cite{km19,dkl18} is to prove a ``fine-grained'' reduction from \textsf{CNF-SAT} to $\CP$ in $\polylog(n)$ dimensions. Fine-grained reductions are reductions between computational problems that keep track of the exact polynomial exponents. For instance, \cite{km19} showed that \textsf{CNF-SAT} with $2^{n(1-o(1))}$ time is reducible to $\CP$ in $\polylog n$ dimensions with $n^{2-o(1)}$ time, and thus the lower bound for \textsf{CP} in $\polylog n$ dimensions is $n^{2-o(1)}$ unless SETH is false.

Surprisingly, to our knowledge, the quantum time complexity of $\CP$ was hardly investigated before. The trivial quantum algorithm for $\CP$ is to use Grover's search algorithm on all $n^2$ pairs, which takes $O(nd)$ time. Sadakane et al.~\cite{SST01} sketched a quantum algorithm that runs in  $O(n^{1-1/(4\ceil{d/2})})$ time. Volpato and Moura~\cite{vm10} claimed a quantum algorithm that uses $O(n^{2/3})$ \emph{queries}, but no analysis was given of the \emph{running time}, and as we will see, the conversion from the query-efficient algorithm to a time-efficient algorithm is nontrivial. As for the lower bound, any quantum algorithm for $\CP$ needs $\Omega(n^{2/3})$ time, since Aaronson and Shi~\cite{AS04} proved such a lower bound for element distinctness, and $\CP$ contains element distinctness as a special case, where a closest pair has distance $0$.

In this work, we resolve the quantum time complexity of $\CP$. In constant dimensions, we observe that by using a quantum walk for element distinctness~\cite{ambainis07,MNRS11}, we can achieve $O(n^{2/3})$ queries for $\CP$.  However, to obtain the same time complexity, the algorithm needs some geometric data structure that supports fast updates and checking, and that---crucially---is ``history-independent'', i.e., the data structure is uniquely represented, disregarding the order of insertion and deletion.
History-independence is essential since different representations of the same data would destroy quantum interference between basis states. 

We propose a geometric data structure that is history-independent and that supports fast checking and updates. Our data structure works by discretizing $\mathbb{R}^d$ into hypercubes with length $\epsilon/\sqrt{d}$. Then, we use a hash table, skip lists, and a radix tree to maintain the locations of the points and hypercubes. This data structure is history-independent, and we can easily find pairs with distance at most $\epsilon$ with it.  We then find the closest pair by a binary search. By using our data structure and a quantum walk~\cite{ambainis07,MNRS11}, we achieve quantum time complexity $\widetilde{O}(n^{2/3})$.   

For $\CP$ in $\polylog(n)$ dimensions, one may expect a conditional lower bound under SETH. However, SETH fails when quantum algorithms are considered since a simple application of Grover's search algorithm on all assignments solves \textsf{CNF-SAT} in time $\tilde{O}(2^{n/2})$. Furthermore, existing fine-grained reductions may require time greater than $O(2^{n/2})$. 

In this paper, we introduce the \emph{Quantum Strong Exponential Time Hypothesis (QSETH)} and \emph{quantum fine-grained reductions}. We define QSETH as
follows. 
\begin{definition}[QSETH]
For all $\epsilon>0$, there exists some $k\in \mathbb{N}$ such that there is no quantum algorithm  solving $k$-$\SAT$ in time $O(2^{(1-\epsilon)\frac{n}{2}})$. 
\end{definition}

We then observe that the classical definition of fine-grained reductions cannot capture the features of quantum reductions such as superposed queries and speedups from quantum algorithms. For instance, a fine-grained reduction may reduce problem $\mathsf{A}$ to solving many instances of problem $\mathsf{B}$ and then output the best solution; in this case, one can use Grover's search algorithm to achieve a quadratic speedup.  Therefore, instead of summing the running time over all instances as in~\cref{def:fg}, we use a quantum algorithm which solves all instances in superposition and outputs the answer. We give a formal definition of quantum fine-grained reductions in~\cref{def:qfgr} and show that under QSETH, any quantum algorithm for $\CP$ in $\polylog (n)$ dimensions requires $n^{1-o(1)}$ time. This implies that Grover's algorithm is optimal for the problem up to an $n^{o(1)}$ factor.  

Intuitively,  QSETH is the conjecture that applying Grover's search algorithm over all assignments in superposition is the optimal quantum algorithm for \textsf{CNF-SAT}. This is similar to SETH, which says that a brute force search is optimal for \textsf{CNF-SAT}.  A  series of works on \textsf{CNF-SAT} \cite{sch90,ppsz05,pp10,her15,ss17} shows that for some constant $c \in [1, 2]$, there exist (randomized) algorithms for $n$-variable $k$-\textsf{SAT} that run in time $2^{n(1-c/k)}$. As $k$ grows, the running time of these algorithms approach $2^n$. When $k$ is small, however, there are algorithms with better running times. For instance, when $k=3$, Sch\"oning \cite{sch90} obtained an algorithm  with $O(1.334^n)$ running time, which was later improved to $O(1.308^n)$ by Paturi et al.~\cite{ppsz05}. However, none of the above mentioned algorithms have good running time on larger $k$'s,  so SETH remains a plausible conjecture.

When $k$ is small enough, there are also quantum algorithms for $k$-\textsf{SAT}~\cite{Ambainis04,DKW05} running in time much less than $O(2^{n/2})$. However, these quantum algorithms mainly use Grover search to speed up the classical algorithms of~\cite{sch90,ppsz05}, and thus do not perform well for large $k$, either. Therefore, we conjecture that for large enough $k$, no quantum algorithm can do much better than Grover search.

Finally, we study the bichromatic closest pair problem ($\BCP$) and the orthogonal vector problem ($\OV$). Briefly, $\OV$ is to find a pair of vectors that are orthogonal given a set of vectors in $\{0,1\}^d \in \mathbb{R}^d$, and $\BCP$ is, given two sets $A, B$ (representing two colors) of $n$ points in $\mathbb{R}^{d}$, to find the pair $(a,b)$ of minimum distance with $a\in A$ and $b\in B$. 

We can summarize all of our results as follows.

\begin{theorem}[Informal]\label{thm:informal_1}
Assuming QSETH, there is no quantum algorithm running in time $n^{1-o(1)}$ for $\OV$, $\CP$, and $\BCP$ when $d=\polylog(n)$.
\end{theorem}

\begin{theorem}[Informal]
The quantum time complexity of $\CP$ in $O(1)$ dimensions\footnote{We actually give a slightly stronger result: the same time complexities still hold when $d=O\left(\frac{\log \log n}{\log \log \log n}\right)$.} is $\tilde{\Theta}(n^{2/3})$\footnote{The $\tilde{\Theta}$ notation is $\Theta$ with logarithmic factors hidden in both upper and lower bounds.}.
\end{theorem}

\begin{theorem}[Informal]
For any $\delta>0$, there exists a quantum algorithm for $\BCP$ with $\tilde{O}(n^{1-\frac{1}{2d}+\delta})$ running time. There exists a quantum algorithm which solves $(1+\xi)$-approximate $\BCP$ in time $\tilde{O}(\xi^{-d}n^{2/3})$.
\end{theorem}

\begin{theorem}[Informal]
The quantum time complexity of  $\OV$ in $O(1)$ dimensions\footnote{The same time complexities still hold when $d=O(\log \log n)$.} is $\Theta(n^{1/2})$.  
\end{theorem}

\cref{tab:summary} also summarizes what is known about upper and lower bounds on the classical and quantum time complexities of all of these problems. 

\setlength\extrarowheight{4pt}
\begin{table}[]
\centering
\resizebox{\textwidth}{!}{%
\begin{tabular}{|c|l|l|l|l|}
\hline
\multicolumn{1}{|l|}{} & Dimension & \textbf{} & Lower Bound & Upper Bound \\ \hline
\multirow{4}{*}{$\CP$} & \multirow{2}{*}{$\Theta(1)$} & Classical & $\wt{\Omega}(n)$~\cite{Yao89} & $\wt{O}(n)$~\cite{SH75,bs76,KM95} \\ \cline{3-5} 
 &  & Quantum & $\mathbf{\Omega(n^{2/3})}$~\cref{thm:cp_constant} & $\mathbf{\wt{O}(n^{2/3})}$~\cref{cor:main_cp} \\ \cline{2-5} 
 & \multirow{2}{*}{$\polylog n$} & Classical & $n^{2-o(1)}$ (Under SETH)~\cite{km19} & $O(n^{2})$ \\ \cline{3-5} 
 &  & Quantum & $\mathbf{n^{1-o(1)}}$ \textbf{(Under QSETH)}~\cref{thm:qseth} & $\mathbf{\wt{O}(n)}$~\cref{cor:trivial_alg} \\ \hline
\multirow{4}{*}{$\OV$} & \multirow{2}{*}{$\Theta(1)$} & Classical & $\Omega(n)$ & $O(n)$~\cite{CST17} \\ \cline{3-5} 
 &  & Quantum & $\mathbf{\Omega(n^{1/2})}$~\cref{thm:ov_constant} & $\mathbf{O(n^{1/2})}$~\cref{thm:ov_constant} \\ \cline{2-5} 
 & \multirow{2}{*}{$\polylog n$} & Classical & $n^{2-o(1)}$ (Under SETH)~\cite{wil05} & $n^{2-o(1)}$~\cite{AWY15,CW16} \\ \cline{3-5} 
 &  & Quantum & $\mathbf{n^{1-o(1)}}$ \textbf{(Under QSETH)}~\cref{thm:qseth} & $\mathbf{\wt{O}(n)}$~\cref{cor:trivial_alg} \\ \hline
 \multirow{5}{*}{$\BCP$} & \multirow{3}{*}{$\Theta(1)$} & Classical & $\Omega(n)$ & $O\bigl(n^{2-\frac{2}{\ceil{d/2}+1}+\delta}\bigr)$~\cite{AES91} \\ \cline{3-5} 
 &  & Quantum & $\mathbf{\Omega(n^{2/3})}$~\cref{thm:bcp-lo} & \begin{tabular}{@{}l@{}}  $\mathbf{\tilde{O}(n^{1-\frac{1}{2d}+\delta})}$ for $\BCP$~\cref{thm:bcp-exact} \\$\mathbf{\tilde{O}(\xi^{-d}n^{2/3})}$ for $\xiBCP$ \cref{thm:bcp-up} \end{tabular}\\ \cline{2-5} 
 & \multirow{2}{*}{$2^{O(\log^*(n))}$\footnotemark} & Classical & $n^{2-o(1)}$ (Under SETH)~\cite{che18} & $n^{2-o(1)}$~\cite{AWY15,CW16} \\ \cline{3-5} 
 &  & Quantum & $\mathbf{n^{1-o(1)}}$ \textbf{(Under QSETH)}~\cref{thm:qseth_2_bcp} & $\mathbf{\wt{O}(n)}$~\cref{cor:trivial_alg} \\ \hline
\end{tabular}
}
\caption{A summary of our quantum complexity results and comparison to classical results. The bold entries highlight our contributions in this paper.}
\label{tab:summary}
\end{table}
\footnotetext{$\log^* (n):=\log^*(\log n)+1$ for $n>1$ and $\log^* (1) := 0$. Hence, $2^{O(\log^* n)}$ is an extremely slow-growing function.}

\paragraph{Related work}
A recent independent work by Buhrman, Patro and Speelman \cite{buhrman2019quantum} also studied quantum strong exponential time hypothesis. They defined (a variant of) QSETH based on the hardness of testing properties on the set of satisfying assignments of a \textsf{SAT} formula, e.g., the parity of the satisfying assignments. Based on these hardness assumptions extended from the original QSETH, they gave conditional quantum lower bounds for $\OV$, the Proofs of Useful Work \cite{ball2017proofs} and the edit distance problem. In comparison, we formally define the \textit{quantum fine-grained reductions} and prove lower bounds for CP, OV, and BCP under the original form of QSETH by showing the existence of quantum fine-grained reductions from CNF-SAT to the these problems.  

\subsection{Proof overview}

For ease of presentation, some notations and descriptions will be informal here. Formal definitions and proofs will be given in subsequent sections.

We give an optimal (up to a polylogarithmic factor) quantum algorithm that solves $\CP$ for constant dimensions in time $\widetilde{O}(n^{2/3})$. First note that there exists a Johnson graph corresponding to an instance of $\CP$, where each vertex corresponds to a subset of $n^{2/3}$ points of the input of $\CP$, and two vertices are connected when the intersection of the two subsets (they are corresponding to) has size $n^{2/3}-1$. A vertex is marked if the subset it corresponds to contains a pair with distance at most $\epsilon$. Then, the goal is to find a marked vertex on this Johnson graph and use binary search over $\epsilon$ to find the closest pair. Our algorithm for finding a marked vertex is based on the quantum walk search framework by Magniez et al.~\cite{MNRS11}, which can be viewed as the quantum version of the Markov chain search on a graph (in our case, a Johnson graph). The complexity of this quantum walk algorithm is $O(\mathsf{S} + \frac{1}{\sqrt{\varepsilon}}(\frac{1}{\sqrt{\delta}}\mathsf{U} + \mathsf{C}))$, where $\varepsilon$ is the fraction of marked states in the Johnson graph, $\delta$ is its spectral gap, $\mathsf{S}$ is the cost for preparing the algorithm's initial state, $\mathsf{U}$ is the cost for implementing one step of the quantum walk, and $\mathsf{C}$ is the cost for checking the solution. For our Johnson graph, $\varepsilon = n^{-2/3}$ and $\delta = n^{-2/3}$. If we consider only the query complexity, $\mathsf{S} = n^{2/3}$, $\mathsf{U} = O(1)$, and $\mathsf{C} = 0$. However, the time complexity for $\mathsf{C}$ is huge in the straightforward implementation, e.g., storing all points in an array according to the index order, as we need to check all the pairs from the $n^{2/3}$ points, which will kill the quantum speedup. To tackle this, we discretize the space into small hypercubes. With this discretization, it suffices to check $O((\sqrt{d})^d)$ neighbor hypercubes to find a pair with distance at most $\epsilon$. To support the efficient neighborhood search, we need an efficient data structure.

Existing data structures do not meet our need. They either have prohibitive dependence on the dimension, such as $\Omega(n^{\ceil{d/2}})$ time for constructing and storing Voronoi diagrams~\cite{Klee80}, or do not have unique representation (i.e., they are history-dependent), such as fair-split trees and dynamic trees~\cite{bespamyatnikh98}. Note that the requirement of unique representation is due to the fact that different representations of the same data would destroy the interference that quantum computation relies on. To solve this problem, we propose a uniquely represented data structure that can answer queries about $\epsilon$-close pairs and insert/delete points efficiently. This data structure is based on a hash table, skip lists, and a radix tree. With this data structure, $\mathsf{U} = O(\log n)$ and $\mathsf{C} = O(1)$. Hence, we have the desired time complexity (see \cref{sec:one-shot}).  We give another method for solving $\CP$ that only uses a radix tree as the data structure.  With only a radix tree, the algorithm cannot handle cases with multiple solutions, and we need to subsequently reduce the size of the problem until there is at most one solution (see \cref{sec:multi-shots}). These two quantum algorithms have the same time complexity.

Our quantum algorithm for solving approximate $\BCP$ follows the same spirit as that for $\CP$, except that we use a finer discritization of the space (see \cref{sec:bcp-approx}). To solve $\BCP$ exactly, we need a history-independent data structure for nearest-neighbor search, but no such data structure is known. Instead, we adapt the nearest-neighbor search data structure by Clarkson~\cite{Clarkson88} to the quantum algorithm proposed by Buhrman et al.~\cite{bdh01} for element distinctness, which does not require history-independence of the data structure because  in the algorithm of~\cite{bdh01}, no insertions and deletions are performed once the data structure for a set of points is constructed (see \cref{sec:bcp-exact}). Sadakane et al.~\cite{SST01} sketched an algorithm for $\BCP$ with similar ideas and running time, but we give the first rigorous analysis.   

To derive our quantum fine-grained complexity results for $\OV$ and $\CP$ when $d=\polylog n$ under QSETH, we first define quantum fine-grained reductions. In our definition, we consider problems whose input is given in the quantum query model, and allow the reduction to perform superposed queries and run quantum algorithms, e.g., amplitude amplification.  The classical reductions from \textsf{CNF-SAT} to $\CP$~\cite{km19,dkl18} and $\OV$~\cite{WY14} are not ``quantum fine-grained'' under QSETH. These reductions fail because their running time exceeds $2^{n/2(1-\epsilon)}$, which is the conjectured time complexity for \textsf{CNF-SAT} under QSETH. Therefore, we cannot derive from them any non-trivial lower bounds for $\CP$ or $\OV$ based on QSETH.  In the following, we use the advantages of quantum algorithms to make these reductions work. 

There are two main obstacles in ``quantizing'' the fine-grained reductions under QSETH. The first obstacle is that the time cost for preparing the input of the problem we reduce to is already beyond the required running time. For instance, consider the reduction from \textsf{CNF-SAT} to $\OV$. Let $\varphi$ be a \textsf{CNF-SAT} instance on $n$ variables and $m$ clauses. The classical fine-grained reduction divides all $n$ variables into two sets $A$ and $B$ of size $n/2$, and then maps all assignments for variables in $A$ and $B$ to two sets $V_A$ and $V_B$ of $2^{n/2}$ vectors each. It is obvious that the time for writing down $V_A$ and $V_B$ is already $\Theta(2^{n/2})$. Nevertheless, many quantum algorithms achieve sublinear query complexities by querying the input oracle in superposition. Hence, instead of first constructing the input of $\OV$ at once and then running the algorithm, we can simulate it ``on-the-fly'': whenever the $\OV$'s algorithm queries the input oracle with some superposition of indices, we use a quantum subroutine to realize the input oracle by mapping the query indices to the corresponding assignments in \textsf{CNF-SAT}, and then to the corresponding vectors in $V_A$ and $V_B$. This subroutine takes only $O(n)$ time, and therefore the quantum reduction, which has running time $O(n)$ times the running time of the $\OV$ algorithm, is quantum fine-grained.

Another difficulty in quantizing the fine-grained reductions is that  some reduction needs to call the oracle multiple times, and the number of calls exceeds the required running time. However, it is possible to achieve quadratic speedup if these oracle calls are non-adaptive. For the reduction from $\BCP$ to $\CP$, we can reduce a $\BCP$ instance to $n^{1.8+o(1)}\log n$ instances of $\CP$, which is already larger than the conjectured $\Omega(n)$ quantum lower bound of $\BCP$. By further studying the reduction, we find that the solution to $\BCP$ is the minimum of the solutions to the the constructed $\CP$ instances. Therefore, we can use the quantum minimum-finding algorithm to reduce the total time complexity to $\wt{O}(\sqrt{n^{1.8+\epsilon}}\cdot t_{\mathsf{CP}})$, which is enough to show that $\BCP$ is quantum fine-grained reducible to $\CP$.

With the above-mentioned techniques, we quantize the classical fine-grained reductions, and show that \textsf{CNF-SAT}, with conjectured lower bound $\Omega(2^{n/2})$, is quantum fine-grained reducible to $\OV$ and $\CP$ with lower bound $\Omega(n')$\footnote{$n$ is the input size of \textsf{CNF-SAT}, and $n'$ is the input size of $\OV$ and $\CP$.}, when the dimension $d$ is $\polylog(n')$.

\section{Preliminaries}
\begin{definition}[Distance measure]
For any two vectors $a,b\in\R^d$, the distance between them in the $\ell_2$-metric is denoted by $\|a-b\|=\left(\sum_{i=1}^d|a_i-b_i|^2\right)^{1/2}$.
Their distance in the $\ell_0$-metric (Hamming distance) is denoted by 
$\|a-b\|_0=|\{i\in[d]:a_i\neq b_i\}|$, i.e.,
the number of coordinates on which $a$ and $b$ differ.
\end{definition}

\subsection{Quantum query model}

We consider the quantum query model in this work. Let $X:=\{x_1,\dots,x_n\}$ be a set of $n$ input points and $\ora_X$ be the corresponding oracle. We can access the $i$-th data point $x_i$ by making the query 
\begin{align}
    \ket{i}\ket{0} \xrightarrow{\ora_X} \ket{i}\ket{x_i},  
\end{align}
and we can make queries to elements in $X$ in superposition. Note that $\ora_X$ is an unitary transformation in the formula above. Hence, a quantum algorithm with access to $\ora_X$ can be represented as a sequence of unitary transformations. 

Consider a quantum algorithm $\A$ with access to an oracle $\ora$ and a initial state $\ket{0}:= \ket{0}_{Q}\ket{0}_{A}\ket{0}_W$, where the registers $Q$ and $A$ are for the queries and the answers from the oracle, and the register $W$ is the working space which is always hold by $\A$.  Then, we can represent the algorithm as 
\begin{align}
    U_T \ora U_{T-1}\cdots\ora U_1  \ket{0}. 
\end{align}
Let $\ket{\psi}_i= U_i\ora \cdots \ora U_1\ket{0}:= \sum_{i,z} \ket{i}_Q\ket{0}_A\ket{z}_W$ be the state right before applying the $i$-th $\ora$, then
\begin{align}
    \ora \ket{\psi}_i:= \sum_{i,z} \ket{i}_Q\ket{x_i}_A\ket{z}_W.
\end{align}
\subsection{Quantum subroutine for unstructured searching and minimum finding}

\begin{definition}[Unstructured search]
Given a set $P$ of $n$ elements in $\{0,1\}$, decide whether there exists a $1$ in $P$. 
\end{definition}

\begin{theorem}[Grover's search algorithm~\cite{Grover1996,nielsen2002quantum}]\label{thm:grover_algo}
There is a quantum algorithm for unstructured search with running time $O(\sqrt{n})$.
\end{theorem}

By \cref{thm:grover_algo} and BBBV's argument \cite{BBBV97}, the quantum time complexity of unstructured search is $\Theta(\sqrt{n})$. We can also get a $\tilde{O}(\sqrt{n})$ quantum algorithm for minimum finding by combining Grover's search algorithm and binary search.
\begin{theorem}[Quantum minimum finding~\cite{durr1996quantum}]\label{thm:mim}
There is a quantum algorithm that finds from a set of $n$ elements with values in $\mathbb{R}$, the index of the minimum element of the set, with success probability $\frac{1}{2}$ and run time $\tilde{O}(\sqrt{n})$.
\end{theorem}

\subsection{Problem definitions}
In this subsection, we first formally define  $\OV$, $\CP$, and $\BCP$. Then we show the folklore algorithms for $\CP$, $\BCP$, and $\OV$ by Grover's algorithm, which run in time $\tilde{O}(n)$. 
\begin{definition}[Orthogonal Vectors, $\OV$]\label{def:OV}
Given two sets $A,B$ of $n$ vectors in $\mathbb \{0,1\}^d$ as input, find a pair of vectors $a\in A$, $b\in B$ such that
$\ip{a}{b}=0$, where the inner product is taken in $\mathbb{Z}$.\footnote{Our definition is slightly different than some of the literature, for example, \cite{chen2019equivalence}, which is searching among pairs inside one set. Those two definitions are equivalent up to constant in complexities.}
\end{definition}

We denote $\OV$ with input length $n$ and dimension $d$ as $\OV_{n,d}$. We will use this notation when we need to specify the parameters in the following sections.  

\begin{definition}[Closest Pair Problem, $\CP$]
Given a set $P$ of $n$ points in $\mathbb R^d$ and a distance measure $\Delta$, find a pair of distinct points $a,b\in P$ such that $\Delta(a,b)$ is the smallest among all distinct pairs in $P$.
\end{definition}
Similar to $\OV$, we denote $\CP$ with input length $n$ and dimension $d$ as $\CP_{n,d}$. We will use this notation when the parameters in the following sections are required to be specified. Note that in this work, we consider $\Delta(a, b) = \norm{a-b}$ as the distance measure for \textsf{CP} and \textsf{BCP}.

\begin{definition}[Bichromatic Closest Pair Problem, \textsf{BCP}]
Given two sets $A,B$ of $n$ points in $\mathbb R^d$ and a distance measure $\Delta$, find a pair of points $a\in A$, $b\in B$ such that 
\begin{align}
    \Delta(a,b)= \min_{a\in A,b\in B} \Delta(a,b).
\end{align}
\end{definition}
We also define an approximate version of $\BCP$ as follows.
\begin{definition}[$(1+\xi$)-approximate Bichromatic Closest Pair Problem, $\xiBCP$]
Given two sets $A,B$ of $n$ points $\in \mathbb R^d$ and a distance measure $\Delta$, find a pair of points $a\in A$, $b\in B$ such that 
\begin{align}
    \Delta(a,b) \leq (1+\xi)\min_{a\in A,b\in B} \Delta(a,b).
\end{align}
\end{definition}
Same as $\CP$, we use $\BCP_{n,d}$ and $\xiBCP_{n, d}$ to specify the parameters. 

\begin{definition}[Element Distinctness Problem, \textsf{ED}]
  Let $f:[n] \rightarrow [m]$ be a given function. Decide whether there exist distinct $i, j\in[n]$ such that $f(i) = f(j)$.
\end{definition}

For this problem, Ambainis \cite{ambainis07} gave a quantum algorithm with time complexity $\tilde{O}(n^{2/3})$, which matches the lower bound proved by Aaronson and Shi \cite{AS04} up to a polylogarithmic factor.

\begin{theorem}\label{cor:trivial_alg}
There are $\tilde{O}(n)$-time quantum algorithms for $\CP$ and $\BCP$ when $d=O(\poly\log n)$.
\end{theorem}
\begin{proof}
We can solve $\CP$ and $\BCP$ by searching the minimum distance through all pairs by the algorithm of \cref{thm:mim}. There are $O(n^2)$ pairs and checking each pair took $O(d)$ time, so the total running time is $O(nd)$. For $d=O(\poly\log n)$, the time complexity equals to $\tilde{O}(n)$.
\end{proof}

\subsection{Fine-grained complexity}

As we have mentioned earlier in the introduction, a fine-grained reduction from problem $\mathsf{P}$ to $\mathsf{Q}$ with conjectured lower bounds $p(n)$ and $q(n)$, respectively, has the property that if we can improve the $q(n)$ time for $\mathsf{Q}$, then we can also improve the $p(n)$ time for $\mathsf{P}$. We give the formal definition by Williams~\cite{vw15} in below. 

\begin{definition}[Fine-grained reduction, \cite{vw15}]\label{def:fg}
Let $p(n)$ and $q(n)$ be non-decreasing functions of $n$. Problem $\mathsf{P}$ is ($p,q$)-reducible to problem $\mathsf{Q}$, denoted as $(\mathsf{P},p)\leq_{\FG} (\mathsf{Q},q)$, if for every $\epsilon$, there exist $\delta>0$, an algorithm $R$ for solving $\mathsf{P}$ with access to an oracle for $\mathsf{Q}$, a constant $d$, and an integer $k(n)$, such that for every $n\ge 1$, the algorithm $R$ takes any instance of $\mathsf{P}$ of size $n$ and 
\begin{itemize}
    \item $R$ runs in at most $d\cdot (p(n))^{1-\delta}$-time, 
    \item $R$ produces at most $k(n)$ instances of $\mathsf{Q}$ adaptively, that is, the $j$th instance $X_j$ is a function of $\{(X_i,y_i)\}_{1\leq i<j}$ where $X_i$ is the $i$th instance produced and $y_i$ is the answer of the oracle for $\mathsf{Q}$ on instance $X_i$, and 
    \item the sizes $n_i$ of the instances $X_i$ for any choice of oracle answers $y_i$ obeys the inequality
    \begin{align}\label{eq:fgr}
        \sum_{i=1}^{k(n)} (q(n_i))^{1-\epsilon} \leq d\cdot (p(n))^{1-\delta}.
    \end{align}
\end{itemize}
\end{definition}

Let $(\mathsf{P},p)\leq_{\FG} (\mathsf{Q},q)$ for some non-decreasing function $p(n)$ and $q(n)$. If for every $\eps > 0$, we can solve problem $\mathsf{Q}$ in time $q(n)^{1-\epsilon}$ with probability $1$ for all input length $n$, then there exists a $\delta>0$ such that we can solve the problem $\mathsf{P}$ in time $p(n)^{1-\delta}$ 
by \cref{eq:fgr}.

Here are some known results about fine-grained reductions.
\begin{theorem}[\cite{km19,wil05}]\label{thm:classical_fg}
\begin{align}
    (\mathsf{CNF\text{-}SAT}_{n},2^n)\leq_{\FG} (\OV_{n_1,d_1},n_1^2)\leq_{\FG} (\BCP_{n_2,d_2},n_2^2)\leq_{\FG} (\CP_{n_3,d_3},n_3^2),
\end{align}
where $d_1=\Theta(\log n_1)$, $d_2=\Theta(\log n_2)$ and $d_3=(\log  n_3)^{\Omega(1)}$.
\end{theorem}

\begin{remark}
The second reduction from $\OV$ to $\BCP$ has been improved to $d_2 = 2^{O(\log^* n)}$ by Chen~\cite{che18}.
\end{remark}

There are several plausible hypotheses in fine-grained complexity, which can imply conditional hardness results for many interesting problems. We first give the definition of the strong exponential time hypothesis (SETH). 
\begin{hypothesis}[Strong Exponential Time Hypothesis, \textsf{SETH}]
For every $\epsilon > 0$, there exists a $k = k(\epsilon) \in \N$ such that no algorithm can solve $k$-$\SAT$ (i.e., satisfiability on a CNF of width $k$) in $O(2^{(1 - \epsilon)m})$ time where $m$ is the number of variables. Moreover, this holds even when the number of clauses is at most $c(\epsilon) \cdot m$ where $c(\epsilon)$ denotes a constant that depends only on $\epsilon$.
\end{hypothesis}

Another popular conjecture is the orthogonal vector hypothesis (\textsf{OVH}):
\begin{definition}[Orthogonal Vector Hypothesis, \textsf{OVH}]
For every $\epsilon > 0$,  
there exists a $c\ge 1$ such that $\OV_{n,d}$ requires $n^{2-\epsilon}$ time when $d=c\log n$.
\end{definition}

\begin{remark}
Under SETH, we can have the following conclusions from \cref{thm:classical_fg}: 
\begin{itemize}
    \item $\OVH$ is true. 
    \item For all $\epsilon>0$, there exists a $c>0$ such that $\BCP_{n,c\log n}$ cannot be solved by any randomized algorithm in time $O(n^{2-\epsilon})$.
    \item For all $\epsilon>0$, there exists a $c>0$ such that $\CP_{n,(\log n)^c}$ cannot be solved by any randomized algorithm in time $O(n^{2-\epsilon})$.
\end{itemize}
\end{remark}

\subsection{The framework for quantum walk search}
\label{sec:qw-framework}
In this subsection, we review the quantum walk framework for the Markov chain search problem and demonstrate how to use it to solve the element distinctness problem. For simplicity, we use the transition matrix $P$ to refer to a Markov chain, where $P = (p_{xy})_{x, y \in X}$ for $X$ being the state space of $P$ and $p_{xy}$ being the transition probability from $x$ to $y$. An irreducible and ergodic Markov chain has a unique stationary distribution $\pi$, which is also the unique eigenvector of $P$ with eigenvalue 1. Let $M \subseteq X$ be a set of marked elements. In the Markov chain search problem, the objective is to find an $x \in M$. We can perform the following actions: setup, sampling from the $\pi$ with cost $\mathsf{S}$; update, making a transition with cost $\mathsf{U}$, and checking whether the current state is marked or not with cost $\mathsf{C}$. To solve the search problem classically, we perform a random walk as follows. Start from a point sampled from $\pi$ and check if it is marked. If not, make a number of transitions on $P$ until it mixes, and then check again. We then repeat this process until a marked state is found. The cost of this random walk algorithm is $O(\mathsf{S} + \frac{1}{\varepsilon}(\frac{1}{\delta}\mathsf{U} + \mathsf{C}))$, where $\varepsilon:= |M|/|X|$ and $\delta$ is the spectral gap of $P$.

Quantum analogues of random walks, namely, quantum walks, have been developed for solving different problems. In 2003, Ambainis~\cite{ambainis07} proposed a quantum walk algorithm for solving the element distinctness problem. His algorithm also solves the Markov chain search problem on the Johnson graph with cost $O(\mathsf{S} + \frac{1}{\sqrt{\varepsilon}}(\frac{1}{\sqrt{\delta}}\mathsf{U} + \mathsf{C}))$. In 2004, Szegedy~\cite{Szegedy04} gave a quantum walk algorithm for more generalized Markov chains with cost $O(\mathsf{S} + \frac{1}{\sqrt{\varepsilon\delta}}(\mathsf{U} + \mathsf{C}))$. We can view Szegedy's quantum walk as a quantum counterpart of a random walk, where one checks the state after each transition. Szegedy's quantum walk only detects the presence of a marked state, but cannot find one without extra costs. In 2006, Magniez et al.~\cite{MNRS11} proposed a quantum walk search framework that unified the advantages of the quantum walks in~\cite{ambainis07} and~\cite{Szegedy04}. In this quantum walk framework, we can perform the following operations:
\begin{itemize}
  \item \textbf{Setup:} with cost $\mathsf{S}$. preparing the initial state $\ket{\pi} = \frac{1}{\sqrt{|X|}}\sum_{x}\sqrt{\pi_x}\ket{x}$.
  \item \textbf{Update:}  with cost $\mathsf{U}$. applying the transformation $\ket{x}\ket{0} \mapsto \ket{x}\sum_{y\in X}\sqrt{p_{xy}}\ket{y}$.
  \item \textbf{Checking:} with cost $\mathsf{C}$, applying the transformation: $\ket{x} \mapsto \big\{\begin{smallmatrix} -\ket{x} \quad &\text{if $x\in M$} \\ \ket{x} \quad &\text{otherwise.}\end{smallmatrix}$
\end{itemize}
The main result of \cite{MNRS11} is summarized as follows.
\begin{lemma}[\cite{MNRS11}]
  \label{lem:framework}
  Let $P$ be an irreducible and ergodic Markov chain $P$ on $X$. Let $M \subseteq X$ be a subset of marked elements. Let $\varepsilon: = |M|/|X|$ and $\delta$ be the spectral gap of $P$. Then, there exists a quantum algorithm that with high probability, determines $M$ is empty or finds an $x \in M$ with cost $O(\mathsf{S} + \frac{1}{\sqrt{\varepsilon}}(\frac{1}{\sqrt{\delta}}\mathsf{U} + \mathsf{C}))$.
\end{lemma}

To solve the element distinctness problem, we define a Markov chain, following the work~\cite{ambainis07,BJLM13,Jeffery2014}. The state space $X$ is all subsets of $[n]$ with size $r$. The Markov chain is based on the Johnson graph on $X$, where an edge is connecting $S$ and $S'$ if and only if $|S \cap S'| = r-1$. The transition probability on each edge is hence $\frac{1}{r(n-r)}$. A state $S$ is marked when there exist distinct $i, j \in S$ such the $i^{th}$ and the $j^{th}$ items are the same. The Markov chain has spectral gap $\delta \geq 1/r$ (see~\cite{Jeffery2014}) and it is easy to verify that $\varepsilon \geq \binom{n-2}{r-2}/\binom{n}{r} = O(r^2/n^2)$. If we only consider the query complexity, the setup procedure costs $r$ queries, the update procedure costs one query, and the checking procedure does not cost any query. Choosing $r = n^{2/3}$ yields the optimal query complexity $O(n^{2/3})$.

\section{Quantum fine-grained complexity}

In this section, we give the formal definitions of the quantum fine-grained reduction and quantum strong exponential time hypothesis (QSETH). Moreover, we show that under QSETH, for $d=\polylog(n)$, the lower bounds for $\CP_{n,d}$ and $\OV_{n,d}$ are $n^{1-o(1)}$, which nearly matches the upper bounds given in \cref{cor:trivial_alg}.  

\subsection{Quantum fine-grained reduction and QSETH}

QSETH is defined based on the assumption that the best quantum algorithm for \textsf{CNF-SAT} is Grover search when the clause width $k$ is large enough. 
\begin{hypothesis}[QSETH]
For every $\epsilon > 0$, there exists a $k = k(\epsilon) \in \N$ such that no quantum algorithm can solve $k$-$\SAT$ (i.e., satisfiability on a CNF of width $k$) in $O(2^{(1/2 - \epsilon)n})$ time where $n$ is the number of variables. Moreover, this holds even when the number of clauses is at most $c(\epsilon) n$ where $c(\epsilon)$ denotes a constant that depends only on $\epsilon$.
\end{hypothesis}

Obviously, the Grover search can solve \textsf{CNF-SAT} in $\tilde{O}(2^{n/2})$. To the best of the our knowledge, there is no quantum algorithm that can do better than $O(2^{n/2})$ for any $k$.

We recall that in the quantum query model, the input of a problem is given by a quantum oracle. Specifically, let $\mathsf{P}$ be a problem, and $X$ be an instance of $\mathsf{P}$ in the classical setting. Then, in the quantum query model, $X$ will be given by an oracle $\ora_X$. We will denote an algorithm or an oracle $\A$ with access to $\ora_X$ by $\A(\ora_X)$. 

We say $\A_{\epsilon}$ is an $\epsilon$-oracle for problem $\mathsf{P}$, if for every instance $\ora_X$, it holds that
\begin{align}
    \Pr[\A_{\epsilon}(\ora_X) = \mathsf{P}(X)] \geq 1-\epsilon,
\end{align}
and the running time is $O(1)$, where $\mathsf{P}(X)$ is the answer of $X$ for problem $\mathsf{P}$. 

\begin{definition}[Quantum oracles]\label{def:oisgood}
Let $X:=\{x_1,\dots,x_n\}$ be an instance of some problem and $\ora_X$ be the corresponding quantum oracle. To realize $\ora_X$, we do not need to write down the whole $X$; instead, we can just design a quantum circuit to realize the mapping 
\begin{align}
    \ket{i}\ket{0} \xrightarrow{\ora_X} \ket{i}\ket{x_i}.
\end{align}  
\end{definition}

\begin{definition}[Quantum fine-grained reduction]\label{def:qfgr}
Let $p(n)$ and $q(n)$ be nondecreasing functions of $n$. Let $\mathsf{P}$ and $\mathsf{Q}$ be two problems in the quantum query model and $\A_{\epsilon}$ be an $\epsilon$-oracle for $\mathsf{Q}$ with error probability $\epsilon\leq 1/3$. $\mathsf{P}$ is quantum ($p,q$)-reducible to $\mathsf{Q}$, denoted as $(\mathsf{P},p)\leq_{\QFG} (\mathsf{Q},q)$, if for every $\epsilon$, there exits a $\delta>0$, and algorithm $R$ with access to $\A_{\epsilon}$, a constant $d$, and an integer $k(n)$, such that for every $n\ge 1$, the algorithm $R$ takes any instance of $\mathsf{P}$ of size $n$ and satisfies the following: 
\begin{itemize}
    \item $R$ can solve $\mathsf{P}$ with success probability at least $2/3$ in time at most $d\cdot p(n)^{1-\delta}$.
    \item $R$ performs at most $k(n)$ quantum queries to $A_\epsilon$. Specifically, in the $j^{th}$ query, let $\mathbf{X}_j:= \{X_{1,j},X_{2,j},\dots\}$ be a set instances of $\mathsf{Q}$. Then, $R$ realizes the oracles $\{\ora_{X_{1,j}},\ora_{X_{2,j}},\dots\}$ in superposition and applies $A_{\eps}$ to solve the instances.
    \item The following inequality holds.
    \begin{align*}
        \sum_{j=1}^{k(n)}c(\mathbf{X}_j)\cdot q(n_j)^{1-\epsilon} \leq d\cdot p(n)^{1-\delta}, 
    \end{align*}
    where $c(\mathbf{X}_j)$ is the time required for $R$ to realize the oracles $\{\ora_{X_{1,j}},\ora_{X_{2,j}},\dots\}$ in superposition and $n_j := \max_{i}|X_{i,j}|$. 
\end{itemize}
\end{definition}

In \cref{def:qfgr}, the input of $A_{\eps}$ is given as a quantum oracle such that $A_{\eps}$ can be a quantum query algorithm with running time strictly less than the input size. Moreover, the quantum reduction $R$ can realize quantum oracles $\{\ora_{X_{1,j}},\ora_{X_{2,j}},\dots\}$ in superposition, and thus the time required is $\max_i c(X_{i,j})$ (where $c(X_{i,j})$ is the time required to realize $\ora_{X_{i,j}}$) instead of $\sum_{i} c(X_{i,j})$. This also allows $R$ to use fast quantum algorithms to process the information of $A_{\epsilon}'s$ output (e.g., amplitude amplification). 

\subsection{Lower bounds for $\CP$, $\OV$, and $\BCP$ in higher dimensions under QSETH}

Here, we give nearly linear lower bounds for $\OV$ and $\CP$ under QSETH by showing that there exist quantum fine-grained reductions from $\textsf{SAT}$ to these problems. 
\begin{theorem}\label{thm:qseth}
Assuming QSETH, for all $\epsilon>0$, there exists a $c$ such that $\OV_{n,c\log n}$ and $\CP_{n,(\log n)^c}$ cannot be solved by any quantum algorithm in time $O(n^{1-\epsilon})$.
\end{theorem}

We prove \cref{thm:qseth} by showing that there exist quantum fine-grained reductions from \textsf{CNF-SAT} to $\OV$, $\OV$ to $\BCP$, and $\BCP$ to $\CP$ with desired parameters.  We first give the reduction from \textsf{CNF-SAT} to $\OV$ as a warm-up. 

\begin{lemma}\label{lem:sat_q_ov}
\begin{align}
    (\mathsf{CNF\text{-}SAT}_{n},2^{n/2})\leq_{\QFG} (\OV_{n_1,d_1},n_1),
\end{align}
where $n_1=2^{n/2}$ and $d_1=\Theta(n)$.
\end{lemma}
\begin{proof}
Let $\phi$ be a CNF formula with $n$ variables and $m=\Theta(n)$ clauses. Let $\mathcal{A}$ be an algorithm for $\OV$. We first recall the classical reduction. Let $\phi := \phi_1\wedge \cdots \wedge\phi_m$. 
We divide the $n$ variables into two sets $A$ and $B$ with $|A| = |B| = \frac{n}{2}$. Let $A:=\{x_1,\dots,x_{n/2}\}$ and $B:= \{x_{n/2+1},\dots,x_{n}\}$. We let $S_A:=\{a_1,\dots,a_{2^{n/2}}\}$ be all assignments to $A$ and $S_B:=\{b_1,\dots,b_{2^{n/2}}\}$ be all assignments to $B$. We describe two mappings $f_A: S_A\rightarrow \{0,1\}^m$ and $f_B:S_B\rightarrow \{0,1\}^m$ as follows: 
\begin{align}
    &f_A(a_i) = [\phi_1(a_i),\dots,\phi_m(a_i)]^T, \text{ and}\\
    &f_B(b_i) = [\phi_1(b_i),\dots,\phi_m(b_i)]^T,
\end{align}
where $\phi_j(a_i) = 0$ if $a_i$ is a satisfied assignment for $\phi_j$, and  $\phi_j(a_i) = 1$ otherwise; we define $\phi_i(b_i)$ in the same way. Let $F_A:= \{f_A(a_i):\; i\in[2^{n/2}]\}$ and $F_B:= \{f_B(b_i):\; i\in[2^{n/2}]\}$. Then, it is obvious that if there exist $v\in F_A$ and $u\in F_B$ such that $\langle v, u \rangle=0$, then $\phi$ is satisfiable. However, at first glance, this reduction with $O(2^{n/2})$ running time is not fine-grained since we require the cost of the reduction to be at most $2^{n(1-\delta)/2}$ for some $\delta>0$ by \cref{def:qfgr}, but writing down elements in $F_A$ and $F_B$ already takes $\Omega(2^{n/2})$. 

Nevertheless, as in \cref{def:oisgood}, a quantum fine-grained reduction only needs to realize the functions $f_A$ and $f_B$, which takes $O(mkn)$ time where $k$ is the width of clauses. This is much less than $O(2^{n(1-\delta)/2})$. More specifically,  $f_A$ and $f_B$ are oracles for $F_A$ and $F_B$, and for any quantum query to elements in $F_A$ or $F_B$, the reduction can implement oracles $f_A$ and $f_B$:   
\begin{align}
    \ket{e,x}\ket{0}\xrightarrow{f_e} \ket{e,x}\ket{f_e(x)},
\end{align}
where $e\in \{A,B\}$, and the time $c(f_e)$ for the reduction to implement $f_e$ for one quantum query is at most $O(kmn)$. Finally, this reduction only uses one oracle ($F_A,F_B$). If there is an algorithm for $\OV$ which succeeds with probability $2/3$, we can boost the success probability of the reduction by repetition.  Therefore, ($\textsf{CNF-SAT},2^{n/2}$) is quantum reducible to ($\textsf{OV}_{n_1,d_1},n_1$). 
\end{proof}

Then, to prove $(\textsf{CNF-SAT}, 2^{n/2})\leq_{\QFG} (\textsf{CP}_{n_3,d_3}, n_3)$, we show that $(\textsf{BCP}_{n_2,d_2},n_2)\leq_{\QFG} (\CP_{n_3,d_3},n_3)$ and $(\textsf{OV}_{n_1,d_1}, n_1) \leq_{\QFG} (\textsf{BCP}_{n_2,d_2},n_2)$, where $n_2, n_3, d_2,d_3$ are some functions of $n$ specified in the following lemmas.

\begin{lemma}\label{lem:qbcp_to_cp}
For $d=\Theta(\log n)$,
\begin{align}
    (\BCP_{n,d},n)\leq_{\QFG} (\CP_{n',d'},n'),
\end{align}
where $n' = n^{O(1)}$ and $d'=(\log n)^{c}$ for some constant $c$ and all points have $\{0,1\}$ entries with the Hamming metric.
\end{lemma}

\begin{remark}
The points have coordinate entries in $\{0,1\}$, and the Hamming metric is equivalent to distance in $\ell_2$-metric (up to power of 2) in this case. Therefore, in the proof of Lemma~\ref{lem:qbcp_to_cp}, we can consider the Hamming distance between points instead of $\ell_2$ distance without loss of generality. 
\end{remark}

We first introduce the classical reductions in~\cite{km19} and some results we will use to prove Lemma~\ref{lem:qbcp_to_cp}. 

\paragraph{Classical reduction} 
We can consider an instance of $\BCP$ with two sets of points $A$ and $B$ as a weighted complete bipartite graph $K_{n,n}$, where the vertices are the points in these two sets and edges' weights are equal to the distances between the corresponding points. Then, solving $\BCP$ is equivalent to find an edge with the minimum weight in this graph.  
However, we cannot directly apply the algorithm for $\CP$ on this graph since there could be two points in the same set (no edge connecting them) that have a smaller distance than any pairs of points in two sets (connected by an edge). To overcome this difficulty, we can ``stretch'' the points to make the points in the same set far from each other, which is characterized by the contact dimension of a graph:
\begin{definition}[Contact Dimension]\label{def:contact_dim}
For any graph $G=(V,E)$, a mapping $\tau:V\rightarrow \R^d$ is said to realize $G$ if for some $\beta>0$, the following holds for every distinct vertices $u,v$:
\begin{align}
    \|\tau(u)-\tau(v)\|_2  =&~\beta ~ \text{if~} \{u,v\}\in E,\\
    \|\tau(u)-\tau(v)\|_2  > &~\beta ~ \text{otherwise.}\notag
\end{align}
The contact dimension of $G$, denoted by $\mathrm{cd}(G)$, is the minimum $d\in \N$ such that there exists $\tau:V\rightarrow \R^d$ realizing $G$.
\end{definition}
That is, with the help of $\tau$, we can restrict the optimal solution of $\CP$ to be the points connected by an edge in $G$. But we cannot realize the whole complete bipartite graph since $\mathrm{cd}(K_{n,n})=\Theta(n)$, which makes the dimension of the $\CP$ instance too large. \cite{km19} showed that we can realize a subgraph of $K_{n,n}$ and apply permutations to its vertices such that the union of these subgraphs cover $K_{n,n}$. In this way, $\BCP$ can be computed by solving $\CP$ on each subgraph and outputting the best solution. More specifically, the reduction in \cite{km19} relies on the following theorem:

\begin{theorem}[Theorem 4.2 in \cite{km19}] \label{thm:dense-cd}
For every $0<\delta < 1$, there exists a log-dense sequence $(n_i)_{i\in \N}$ such that, for every $i\in \N$, there is a bipartite graph $G_i=(A_i \dot \cup B_i, E_i)$ where $|A_i|=|B_i|=n_i$ and $|E_i|\ge \Omega(n_i^{2-\delta})$, such that $\mathrm{cd}(G_i)=(\log n_i)^{O(1/\delta)}$. Moreover, for all $i\in \N$, a realization $\tau:A_i\dot \cup B_i\rightarrow \{0,1\}^{(\log n_i)^{O(1/\delta)}}$ of $G_i$ can be constructed in $n_i^{2+o(1)}$ time.
\end{theorem}
\noindent The log-dense sequence is defined as follows:
\begin{definition}
A sequence $(n_i)_{i\in \N}$ of increasing positive integers is log-dense if there exists a constant $c\ge 1$ such that $\log n_{i+1}\le c\cdot \log n_i$ for all $i\in \N$.
\end{definition}
They also showed that, the permutations for covering the complete bipartite graph can be efficiently found, as shown in the following lemma.
\begin{lemma}[Lemma 3.11 in \cite{km19}] \label{lem:cover}
For any bipartite graph $G(A\dot \cup B, E_G)$ where $|A|=|B|=n$ and $E_G\ne \emptyset$, there exist side-preserving permutations $\pi_1,\dots, \pi_k:A\cup B\rightarrow A\cup B$ where $k\leq \frac{2n^2\ln n}{|E_G|}+1$ such that
\begin{align}
    \underset{i\in [k]}{\bigcup} E_{G_{\pi_i}} = E_{K_{n,n}}.
\end{align}
Moreover, such permutations can be found in $O(n^6\log n)$ time.
\end{lemma}

Now, we are ready to state the quantum fine-grained reduction by ``quantizing'' the classical reduction.
\begin{proof}[Proof of \cref{lem:qbcp_to_cp}]
Let $A,B$ be the two sets of input points of $\textsf{BCP}$. Suppose for \textsf{BCP}, there is an input oracle $\ora_{\BCP}$ which, given an index, returns the corresponding point:
\begin{align}
    \ket{b}\ket{i}\ket{0} \xrightarrow{\ora_{\BCP}}
    \begin{cases}
        \ket{b}\ket{i}\ket{x_i} &\text{if } b= 0,\\
        \ket{b}\ket{i}\ket{y_i} & \text{if } b= 1,
    \end{cases}
\end{align}
where $x_i$ is the $i$-th point in the set $A$ and $y_i$ is the $i$-th point in the set $B$. The sizes of $A$ and $B$ are both equal to $n$ and each point is in $\{0,1\}^{d_1}$, where $d_1=\Theta(\log n)$ is the dimension of $\BCP$. 
 
For $\CP$, suppose there is a quantum algorithm $\mathcal{A}$ such that for $m$ points in $\{0,1\}^{d_2}$ given by an oracle $\mathcal{M}_{CP}$, $\mathcal{A}^{\mathcal{M}_{CP}}$ returns the closest pair of these $n$ points with probability at least $2/3$.

Then we need to transform $\ora_{\BCP}$ to some oracles $\mathcal{M}_i$ for \CP, such that by running $\mathcal{A}$ with $\mathcal{M}_i$ as input oracles, we can get the bichromatic closest pair between $A$ and $B$. The reduction has four steps:

\paragraph*{1. Pre-processing.} We first follow the classical reduction to pre-process the input points of \textsf{BCP}. For some integer $n'\leq n^{0.1}$, we can partition $A$ and $B$ into $n'$-size subsets:
\begin{align}
    A= &~ A_1~\dot\cup ~\cdots ~\dot \cup~ A_{r},\\
    B= &~ B_1 ~\dot\cup ~\cdots ~\dot \cup~ B_{r},\notag
\end{align}
where $r=\lfloor n/n'\rfloor$. Here, we assume that $n$ is divisible by $n'$. It follows that
\begin{align}\label{eq:bcp_bcp}
    \textsf{BCP}(A,B)=\min_{i,j\in [r]} \textsf{BCP}(A_i,B_j).
\end{align}

Then, we use the algorithm in \cite{km19} to construct $k$ mappings $f_1,\dots, f_k:[2n']\rightarrow \{0,1\}^{d'}$ such that 
\begin{align}\label{eq:bcp_cp}
    \textsf{BCP}(A_i,B_j)=\min_{t\in [k]} \textsf{CP}(f_t(A_i)\cup f_t(B_j)) ~~~ \forall i,j\in [\lfloor n/n'\rfloor].
\end{align}

More specifically, we pick $n'$ to be the largest number in a log-dense sequence that is smaller than $n^{0.1}$. Then, we apply~\cref{thm:dense-cd} to classically construct a bipartite graph $G(A\cup B,E)$ with $n'$ vertices in each side and a realization $\tau$.
By choosing $\delta = \epsilon/2$ in \cref{thm:dense-cd}, the graph $G$ has $|E|=\Omega(n'^{2-\epsilon/2})$ edges. And we can get $2n'$ 0/1-strings of length $(\log n')^{O(2/\epsilon)}$:
\begin{align}
    \tau^A_i=\tau(u_i)~~~\forall u_i\in A, \quad\text{ and }\quad
    \tau^B_i=\tau(v_i)~~~\forall v_i\in B.
\end{align}

In order to cover the complete bipartite graph, we run the classical algorithm (\cref{lem:cover}) to find $k$ permutations $\pi_1,\dots, \pi_k:[n']\rightarrow [n']$, where $k$ is a parameter to be chosen later.

Then, we can define the mappings as follows:  
\begin{align}\label{eq:fcn_cp}
    f_t(u)=
    \begin{cases}
        x_v\circ \left(\tau^A_{\pi_t(w)}\right)^{d+1} & \text{if }1\leq u\leq n'\\
        y_v\circ \left(\tau^B_{\pi_t(w)}\right)^{d+1} & \text{if }n'<u\leq 2n'
    \end{cases} ~~~\forall t\in [k], u\in [2n'],
\end{align}
where $\circ$ means string concatenation and $(s)^{d+1}$ denotes $d+1$ copies of the string $s$. For a point $p\in A_i\cup B_j$, $u\in [2n']$ is the index in this union-set,  $v\in [n]$ is the index in the ground set $A$ or $B$, and $w\in [n']$ is the index in the subset $A_i$ or $B_j$. Further, if $1\leq u\leq n'$, then $w:=u$; otherwise, $w:=u-n'$.

\paragraph{2. Oracle construction.}
For $i,j\in [r],t\in [k]$, we then construct the input oracle $\mathcal{M}_{i,j,t}$ for the problem $\textsf{CP}(f_t(A_i)\cup f_t(B_j))$. For a query index $u\in [2n']$, 
\begin{align}\label{eq:oracle_cp}
    M_{i,j,t}\ket{u}\ket{0}=\ket{u}\ket{f_t(u)}.
\end{align}

With the help of the input oracle $\ora_{\BCP}$, we can implement $\mathcal{M}_{i,j,t}$ in the following way: 
\begin{enumerate}
    \item Prepare an ancilla qubit $\ket{b}$ such that $b=1$ if $u>n'$. 
    \item Transform $\ket{u}$ to $\ket{v}$, the index of the point in $A$ or $B$, based on the value of $b$. Note that the index is unique. Hence, this transformation is unitary and can be easily achieved by a small quantum circuit.
    \item Query $\ora_{\BCP}$ with input $\ket{b}\ket{v}$. Assume $b=0$. Then, 
    \begin{align}
        \ket{b}\ket{v}\ket{0}\xmapsto{\ora_{\BCP}}\ket{b}\ket{v}\ket{x_v}.
    \end{align}
    \item Similar to the second step, the index $w$ of the point in $A_i$ and $B_j$ can be computed from $v$ by a unitary transformation:
    \begin{align}
        \ket{b}\ket{v}\ket{x_v}\mapsto \ket{b}\ket{w}\ket{x_v}
    \end{align}
    \item Since each $w$ corresponds to a unique string $\tau^A_{\pi_t(w)}$, we can attach $d+1$ copies of this string to the remaining quantum registers:
    \begin{align}
        \ket{b}\ket{w}\ket{x_v}\mapsto \ket{b}\ket{w}\ket{x_v}\ket{\left( \tau^A_{\pi_t(w)} \right)^{d+1}}.
    \end{align}
    \item By recovering $u$ from $w$, we get the final state:
    \begin{align}
        \ket{u}\ket{f_t(u)}=\ket{u}\ket{x_v,\left( \tau^A_{\pi_t(w)} \right)^{d+1}}.
    \end{align}
\end{enumerate}

\paragraph*{3. Query process}
By \cref{eq:bcp_bcp,eq:bcp_cp}, we have
\begin{align}\label{eq:bcp_all}
    \textsf{BCP}(A,B)=\min_{i,j\in [r],t\in [k]}\textsf{CP}(f_t(A_i)\cup f_t(B_j)).
\end{align}
Hence, we can use quantum minimum-finding algorithm in \cref{cor:trivial_alg} over the sub-problems to find the minimum solution. For each sub-problem, we can run the algorithm for \textsf{CP} with $\mathcal{M}_{i,j,t}$ as the input oracle.

\paragraph*{4. Post-processing.} In case that $n$ is not divisible by $n'$, let the remaining points in $A$ and $B$ be $A_{res}$, $B_{res}$, respectively. Then, we can use Grover search to find the closest pair between $A_{res}$ and $B$, and between $B_{res}$ and $A$. Then, compare the answer to the previously computed result and pick the smaller one.

\paragraph*{Correctness.} In this reduction, we do not change the constructions of the mappings $\{f_i\}_{i\in [k]}$. By~\cite{km19}, \cref{eq:bcp_all} is correct in the classical setting. Hence, it also holds in the quantum setting, and we can use Grover search to find the minimum solution. However, since the algorithm $\mathcal{A}$ for $\CP$  has success probability $2/3$, for each tuple $(i,j,t)\in [r]\times [r]\times [k]$, we need to run $\mathcal{A}^{\mathcal{M}_{i,j,t}}$ $O(\log n)$ times to boost the success probability to at least $1-\frac{1}{n}$. Then, by the union bound, the probability that all queries in the Grover search are correct is at least $99/100$. Hence, by \cref{thm:mim}, the overall success probability is at least $2/3$.

\paragraph*{Running Time of the Reduction.} 

The running time of the pre-processing step consists of two parts: (1) constructing the graph $G$ and its realization $\tau$; (2) finding $k$ permutations. For the first part, by~\cref{thm:dense-cd}, it can be done in $n'^{2+o(1)}$ time. For the second part, we pick $k=O(\frac{2n'^2\log n'}{n'^{2-\epsilon/2}})=O(n'^{\epsilon/2}\log n')$, and by~\cref{lem:cover}, it can be done in $O(n'^6 \log n')$ time. Hence, the total running time of pre-processing step is $n'^{2+o(1)}+O(n'^6\log n')=\wt{O}(n^{0.6})$.

The oracle construction can be done ``on-the-fly''. More specifically, given the strings $\{\tau^A_i,\tau^B_i\}_{i\in [n']}$, and permutations $\{\pi_i\}_{i\in [k]}$, for each query index $u$, we can simulate the oracle $\mathcal{M}_{i,j,t}$ defined in \cref{eq:oracle_cp} in $c(\mathcal{M}_{i,j,t})=O(d_2)=(\log n')^{\Omega(1)}=\wt{O}(1)$ time.

In the query process, for each $\CP$ instance indexed by $(i,j,t)$, suppose $\mathcal{A}^{\mathcal{M}_{i,j,t}}$ gets the answer in time $q(n')=n'$. Moreover, for each time $\mathcal{A}$ querying the input oracle $\mathcal{M}_{i,j,t}$, we need to spend $c(\mathcal{M}_{i,j,t})$ time to simulate the oracle. And we also have $O(\log n)$ runs for each instance. Hence, the total running time for each $\CP$ is at most  
\begin{align}
    n'^{1-\epsilon}\cdot \wt{O}(1)\cdot O(\log n)=\wt{O}(n'^{1-\epsilon}).
\end{align}
Then, we use Grover's search algorithm over $r^2\cdot k$ instances, which can be done by querying $\wt{O}(\sqrt{r^2\cdot k})$ instances by \cref{thm:mim}. Therefore, for any $\epsilon > 0$, we have

\begin{align}\label{eq:bcp_qfgr}
    \wt{O}(\sqrt{r^2 k})~\cdot &~ q(n')^{1-\epsilon}\cdot c(\mathcal{M}_{i,j,t})\cdot O(\log n) =  ~ 
    \wt{O}(\sqrt{(n/n')^2 k}\cdot(n')^{1-\eps}) \\
    \leq &~ \wt{O}(n\cdot (n')^{-\eps})
    \leq \wt{O}(n\cdot n^{-\eps/2})
    \leq n^{1-\delta},
\end{align}
where the first inequality follows from $k=O(n'^{\epsilon / 2}\log n')$ as shown in \cite{km19} and the last inequality follows by setting $\delta = \epsilon / 10$.

For the post-processing step, the sizes of $A_{res}$ and $B_{res}$ are at most $n'$. The running time is 
\begin{align}
    O(\sqrt{n\cdot n'}\cdot \log n) \leq \wt{O}(n^{0.55}).
\end{align}

Therefore, for any $\epsilon>0$, there exists a $\delta>0$ such that the \cref{eq:bcp_qfgr} holds and the total reduction time is $O(n^{1-\delta})$. By \cref{def:qfgr}, $\BCP_{n,d_1}$ can be quantum fine-grained reduced to $\CP_{n,d_2}$. This completes the proof of this lemma. 
\end{proof}

Finally, we show that $(\textsf{OV}_{n,d},n)\leq_{\QFG} (\textsf{BCP}_{n,d'},n)$ by quantizing the reduction in~\cite{km19} following the same idea. 
\begin{lemma}\label{lem:ov_q_bcp}
For $d=\Theta(\log n)$, 
\begin{align} 
    (\OV_{n,d},n)\leq_{\QFG} (\BCP_{n,d'},n),
\end{align}
where $d'=\Theta(\log n)$.
\end{lemma}
\begin{proof}
For an $\OV$ instance with sets of vectors $A$ and $B$, let $\ora_{\OV}$ be the input oracle such that:
\begin{align}
    \ora_{\OV}\ket{i}\ket{0}=\begin{cases}
        \ket{i}\ket{a_i} & \text{if }i\in A,\\
        \ket{i}\ket{b_i} & \text{if }i\in B.
    \end{cases}
\end{align}
where $a_i,b_i\in \{0,1\}^d$.

Then, similar to the classical reduction, we can construct mappings $f_A, f_B:\{0,1\}^d\rightarrow \{0,1\}^{5d}$ such that
\begin{align}
    f_A(a_i)_{5j-4:5j}=\begin{cases}
        11000 & \text{if }a_i(j)=0\\
        00110 & \text{if }a_i(j)=1
    \end{cases}~~~\forall j\in [d],
\end{align}
and
\begin{align}
    f_B(b_i)_{5j-4:5j}=\begin{cases}
        10100 & \text{if }b_i(j)=0,\\
        01001 & \text{if }b_i(j)=1.
    \end{cases} ~~~\forall j\in [d].
\end{align}

By the classical reduction, we have
\begin{align}
    \textsf{OV}(A,B)=1 \text{ if and only if }\textsf{BCP}(f_A(A),f_B(B))=2d
\end{align}
under Hamming distance. 

Also, note that we can simulate the input oracle $\ora_{\BCP}$ by first querying the oracle $\ora_{\OV}$ to get the vector, then applying the corresponding mapping $f_A$ or $f_B$, which can be done in $c(\ora_{\BCP})=O(d)$ time. Let the running time of the algorithm for $\textsf{BCP}$ be $q(n)=n$. Then for any $\epsilon>0$,
\begin{align}
    q(n)^{1-\epsilon}\cdot c(\ora_{\BCP})=n^{1-\epsilon}\cdot \Theta(\log n)\leq n^{1-\delta}
\end{align}
for some small $\delta >0$. Hence, by \cref{def:qfgr}, $(\textsf{OV}_{n,d},n)\leq_{\QFG} (\textsf{BCP}_{n,d'},n)$.
\end{proof}
 
\begin{proof}[Proof of \cref{thm:qseth}]

We can prove the theorem by contradiction following \cref{lem:sat_q_ov}, \cref{lem:ov_q_bcp}, and \cref{lem:qbcp_to_cp}. Specifically, suppose that there exists an $\epsilon>0$, for all $d=\Theta(\log n)$, there exists a quantum algorithm which can solve $\OV$ in time $O(n^{1-\epsilon})$. Then, we can obtain a quantum algorithm for \textsf{CNF-SAT}, which runs in time $O(2^{n/2(1-\epsilon)})$ by \cref{lem:sat_q_ov}. This contradicts QSETH. The proof for $\CP$ is the same. 
\end{proof}

\subsection{Quantum lower bound for $\BCP$ in nearly-constant dimensions under QSETH}
A byproduct of the previous subsection is a quantum lower bound for $\BCP$ in higher dimensions (i.e., $d = \polylog(n)$) under QSETH (\cref{lem:ov_q_bcp}). In this subsection, we show that this quantum lower bound for $\BCP$ even holds for nearly-constant dimensions (i.e., $d = c^{\log^*(n)}$). The main result of this subsection is the following theorem.
\begin{theorem}\label{thm:qseth_2_bcp}
Assuming QSETH, there is a constant $c$ such that $\BCP$ in $c^{\log^* (n)}$ dimensions requires $n^{1-o(1)}$ time for any quantum algorithm.
\end{theorem}

We will ``quantize'' the results by Chen~\cite{che18} to prove this theorem. More specifically, we first show a quantum fine-grained self-reduction of $\OV$ from $\log n$ dimensions with binary entries to $2^{\log ^* (n)}$ dimensions with integer entries ($\ZOV$). Then, we give a quantum fine-grained reduction from $\ZOV$ to $\BCP$ in nearly-constant dimensions.

\begin{definition}[Integral Orthogonal Vector, $\ZOV$]
Given two sets $A,B$ of $n$ vectors in $\mathbb \Z^d$, find a pair of vectors $a\in A$ and $b\in B$ such that
$\ip{a}{b}=0$, where the inner product is taken in $\mathbb{Z}$.
\end{definition}

We use $\ZOV_{n,d}$ to denote $\ZOV$  with $n$ vectors of $d$ dimension in each set. We then recap a theorem in \cite{che18}:

\begin{theorem}[{\cite[Theorem 4.1]{che18}}]\label{thm:zov_map}
Let $b,\ell$ be two sufficiently large integers. There is a classical reduction $\psi_{b,\ell}:\{0,1\}^{b\cdot\ell}\rightarrow \Z^\ell$ and a set $V_{b,\ell}\subseteq \Z$, such that for every $x,y\in \{0,1\}^{b\cdot \ell}$,
\begin{align}\label{eq:zov_iff}
    \ip{x}{y}=0~\Leftrightarrow~\ip{\psi_{b,\ell}(x)}{\psi_{b,\ell}(y)}\in V_{b,\ell}
\end{align}
and
\begin{align}
    0\leq \psi_{b,\ell}(x)_i< \ell^{6^{\log^* (b)}\cdot b}
\end{align}
for all possible $x$ and $i\in [\ell]$. Moreover, the computation of $\psi_{b,\ell}(x)$ takes $\poly(b\cdot \ell)$ time, and the set $V_{b,\ell}$ can be constructed in $O\left( \ell^{O(6^{\log^* (b)}\cdot b)}\cdot \poly(b\cdot \ell) \right)$ time.
\end{theorem}

Note that the size of $V_{b,\ell}$ is at most $\ell^{2\cdot 6^{\log^* (b)}\cdot b+1}$. 
The following lemma gives a quantum fine-grained reduction from $\OV$ to $\ZOV$:
\begin{lemma}\label{lem:ov_2_zov}
For $d=\Theta(\log n)$,
\begin{align}
    (\OV_{n,d}, n)\leq_{\QFG} (\ZOV_{n,d'}, n).
\end{align}
where $d'=2^{O(\log^* n_2)}$.
\end{lemma}
\begin{proof}
Consider an $\OV_{n,d}$ with $d=c\cdot \log n$, where $c$ is an arbitrary constant. We choose $\ell:=7^{\log^* n}$ and $b:=d/\ell$. Then, we can apply~\cref{thm:zov_map} to get the mapping function $\psi_{b,\ell}$ and the set $V_{b,\ell}$. For each $v\in V_{b,\ell}$, we'll construct an instance of $\ZOV_{n,\ell+1}$ as follows:
\begin{enumerate}
    \item Let $\ket{i}$ be the input query index of $\ZOV_{n,\ell+1}$.  
    \item Query $\OV_{n,d}$'s input oracle $\ora_{\OV}$ and get the vector $\ket{i,x}$.
    \item Compute the mapping $\psi_{b,\ell}$ and get $\ket{i,x}\ket{\psi_{b,\ell}(x)}$.
    \item If $x\in A$, then attach 1 to the end of the register: $\ket{i,x}\ket{\psi_{b,\ell}(x),1}$. If $x\in B$, then attach $-v$ to the end: $\ket{i,x}\ket{\psi_{b,\ell}(x),-v}$.
    \item Use $\ora_{\OV}$ to erase $x$ and return the final input state $\ket{i}\ket{\psi_{b,\ell}(x),1}$ or $\ket{i}\ket{\psi_{b,\ell}(x),-v}$.
\end{enumerate}
For each instance, we can use the quantum oracle for $\ZOV_{n,\ell+1}$ to check the orthogonality. $\OV_{n,d}$ is YES if and only if there exists a YES-instance of $\ZOV_{n,\ell+1}$.

\paragraph*{Correctness.} The correctness follows from~\cref{eq:zov_iff}: 
\begin{align}
\ip{x}{y}=0 \Leftrightarrow \ip{\psi_{b,\ell}(x)}{\psi_{b,\ell}(y)}=v\in V_{b,\ell} \Leftrightarrow \ip{[\psi_{b , \ell}(x), ~ 1]}{[\psi_{b,\ell}(y), ~ -v]}=0. 
\end{align}

\paragraph*{Reduction time.} Note that for $\ell=7^{\log^* n}$ and $b=d/\ell$, we have:
\begin{align}
    \log\left( \ell^{O(6^{\log^* (d)}\cdot b)} \right) = &~ \log \ell \cdot O\left( 6^{\log^* (d)}\cdot (d/\ell) \right)\\
    =&~ O\left( \log^* (n)\cdot 6^{\log^* n} \cdot c\log n / 7^{\log^* n} \right)\\
    =&~ o(\log n).
\end{align}
This implies that $|V_{b,\ell}|\leq \ell^{2\cdot 6^{\log^* (b)}\cdot b+1} \leq n^{o(1)}$.
Hence, the number of $\ZOV_{n,\ell+1}$ instances is $n^{o(1)}$ and the running time for compute $V_{b,\ell}$ is $n^{o(1)}$. And for each input query, the oracle for $\ZOV_{n,\ell+1}$ can be simulated in $c(\ora_{\ZOV})=\poly(d)=\poly(\log n)$ time. We can show that for every $\epsilon>0$, if $\ZOV_{n,\ell+1}$ can be decided in $n^{1-\epsilon}$ time, then
\begin{align}
    \sum_{v\in V_{b,\ell}} n^{1-\epsilon}\cdot c(\ora_{\ZOV})=n^{o(1)}\cdot n^{1-\epsilon}\cdot \poly(\log n) \leq n^{1-\delta}
\end{align}
for some $\delta > 0$, which satisfies the definition of quantum fine-grained reduction (\cref{def:qfgr}).

Therefore, $\OV_{n,O(\log n)}$ is quantum fine-grained reducible to $\ZOV_{n,2^{O(\log^* (n))}}$.
\end{proof}

Then, we give a quantum fine-grained reduction from $\ZOV$ to $\BCP$:
\begin{lemma}\label{lem:zov_2_bcp}
For $d=2^{O(\log^* n)}$,
\begin{align}
    (\ZOV_{n,d}, n)\leq_{\QFG} (\BCP_{n,d'}).
\end{align}
where $d'=d^2+2$.
\end{lemma}
\begin{proof}
We remark here that this proof closely follows that for Theorem 4.3 in \cite{che18}. We nonetheless give it here as some details are different.

For an $\ZOV_{n,d}$ instance with $(k\cdot \log n)$-bit entries, we construct a $\BCP$ instance as follows:
\begin{enumerate}
    \item For $x\in A$, construct a vector $x'\in \Z^{d^2}$ such that $x'_{i,j}=x_i\cdot x_j$. Here, we index a $d^2$-dimensional vector by $[d]\times [d]$. Similarly, for $y\in B$, construct a vector $y'\in \Z^{d^2}$ such that $y'_{i,j}=-y_i \cdot y_j$.
    \item Choose $W:=(d^2+1)\cdot n^{4k}$. For each $x'$, construct a vector $x''\in \R^{d^2+2}$ such that 
    \begin{align}
        x''=\left[x',~ \sqrt{W-\|x'\|_2^2},~ 0\right].
    \end{align}  For each $y'$, construct a vector $y''\in \R^{d^2+2}$ such that 
    \begin{align}
        y''=\left[y',~ 0, ~ \sqrt{W-\|y'\|_2^2}\right].
    \end{align}
\end{enumerate}
Then, we claim that the $\ZOV$ instance is YES if and only if the $\BCP$ instance has the minimum distance $\leq \sqrt{2W}$.

\paragraph*{Correctness.} First note that $\|x'\|_2^2\leq d^2\cdot (2^{k\log n})^4=d^2\cdot n^{4k}$. Hence, $W-\|x'\|_2^2 > 0$ and $W-\|y'\|_2^2>0$. For any $x''$ and $y''$ in the new constructed instance of $\BCP$, we have
\begin{align}
    \|x''-y''\|_2^2=&~ \|x''\|_2^2+\|y''\|_2^2-2\cdot \ip{x''}{y''}\\
    =&~ 2\cdot W-2\cdot \ip{x'}{y'}\\
    =&~ 2\cdot W - 2\cdot \sum_{(i,j)\in [d]\times [d]}x_i\cdot x_j\cdot (-y_j\cdot y_j)\\
    =&~ 2\cdot W + 2\cdot (\ip{x}{y})^2.
\end{align}
Hence,
\begin{align}
    \ip{x}{y}=0~\Leftrightarrow~\|x''-y''\|_2^2= 2W.
\end{align}

\paragraph*{Reduction time.} We can see from the above description that the input mapping function is simple and can be computed by a small quantum circuit in $O(d^2)=O(2^{O(\log^* (n))})$ time. Hence, we have $c(\ora_{\BCP})=O(2^{O(\log^* (n))})$. Also, by~\cref{def:qfgr}, it's easy to check that this is indeed a quantum fine-grained reduction from $\ZOV$ to $\BCP$.
\end{proof}

Now \cref{thm:qseth_2_bcp} follows immediately from \cref{lem:ov_2_zov} and \cref{lem:zov_2_bcp}:
\begin{proof}[Proof of \cref{thm:qseth_2_bcp}]
Let $\eps > 0$ be some constant. Suppose we can solve $\BCP_{n,c^{\log^* (n)}}$ in $n^{1-\epsilon}$ time for all constant $c>0$. Then, by \cref{lem:ov_2_zov} and \cref{lem:zov_2_bcp}, we can also solve $\OV_{n,c'\log n}$ in $n^{1-\delta}$ time for some $\delta > 0$ and any $c'>0$. However, this contradicts QSETH by~\cref{thm:qseth}. Therefore, assuming QSETH, there exists a constant $c$ such that $\BCP_{n,c^{\log^*( n)}}$ requires $n^{1-o(1)}$ time.
\end{proof}

\section{Closest pair in constant dimension}
\label{sec:cp}

In this section, we show that there exist almost-optimal quantum algorithms for $\CP$ in constant dimension. The main result is the following theorem, which is a direct consequence of \cref{cor:main_cp,thm:cp_constant}.
\begin{theorem}
For any constant dimension, the quantum time complexity for $\CP$ is $\tilde{\Theta}(n^{2/3})$. 
\end{theorem}

Our approach to solve \textsf{CP} is first  reducing to the decision version of the problem, and then apply quantum walk algorithms to solve the decision version. We define the decision version of \textsf{CP},  $\eCP$, as follows. 
\begin{definition}[$\eCP$]
Given a set of points $P\subset \mathbb{R}^{d}$ and $\epsilon\in \mathbb{R}$, find a pair $a,b\in P$ such that $\|a-b\|\leq \epsilon$ if there is one and returns \textsf{no} is no such pair exists. 
\end{definition}
The reduction from $\CP$ to $\eCP$ is given by the following lemma.
\begin{lemma}
  \label{lemma:cp-ecp}
Let $m$ be the number of bits needed to encode each coordinate as a bit string and $d$ be the dimension.
Given an oracle $\ora$ for $\eCP$, there exists an algorithm $A^{\ora}$  that runs in time and query complexity $O(m+\log d )$ that solves the $\CP$. 
\end{lemma}
\begin{proof}
Let $(P,\delta)$ be an instance of the $\CP$.  We first pick an arbitrary pair $a_0,b_0\in P$ and compute $\Delta(a_0,b_0)$. Then, we set $\epsilon$ to be $\Delta(a_0,b_0)/2$ and run the oracle $\ora$ to check whether there exists a distinct pair with distance less than $\Delta(a_0,b_0)/2$ or not. If there exists such a pair, which we denote as $(a_1,b_1)$, then we set $\epsilon=\Delta(a_1,b_1)$ and call $\ora$ to check again. If there is no such pair, then we set $\epsilon=3\Delta(a_0,b_0)/4$ and call $\ora$. We run this binary search for $m+\log d$ iterations. Finally, the algorithm outputs the closest pair.  
\end{proof}

In classical setting, point location is an important step in solving the closest-pair problem, especially the dynamic version. For the quantum algorithm, as walking on the Markov chain, we repeatedly delete a point and add a new point. Hence, in each step, the first thing is to determine the location of the new added point.    

For simplicity, we assume that $m=O(\log n)$, which is the number of digits of each coordinate of the points. By translation, we can further assume that all the points are lying in $[0,L]^d$, where $L=O(2^m) = \poly(n)$. 

Since we are considering $\eCP$, one simple way of point location is to discretize the whole space into a hypergrid, which is defined as follows: 

\begin{definition}\label{def:hypergrid}
Let $d, \epsilon, L>0$. A hypergrid $G_{d,\epsilon,L}$ in the space $[0,L]^d$ consists of all  $\eps$-boxes 
\begin{align}
    g := [a_1,b_1)\times [a_2,b_2)\times \cdots \times [a_d,b_d), 
\end{align}
such that
$b_1-a_1=\cdots =b_d-a_d=\epsilon/\sqrt{d}$ \footnote{The diagonal length of an $\eps$-box is $\epsilon$.}, and $a_i$ is divisible by $\epsilon$ for all $i\in [d]$. 
\end{definition}

For each point $p_i\in [0,L]^d$, we can identify the $\epsilon$-box that contains it using the function $\id(p_i):[0,L]^d\rightarrow \{0,1\}^{d\log (L/\epsilon)}$:
\begin{align}\label{eq:id_fcn}
    \id(p_i)=\big(\left\lfloor p_i(1)/w\right\rfloor, \left\lfloor p_i(2)/w\right\rfloor, \dots, \left\lfloor p_i(d)/w\right\rfloor\big),
\end{align}
where $w=\frac{\epsilon}{\sqrt{d}}$ is the width of the $\eps$-box. The number of bits to store $\id(p_i)$ is $d\cdot \log(L/w)=O(d\cdot \log(L))$.
Since all the points in an $\epsilon$-box have the same $\id$, we also use this $g(\id(p))$ to denote this $\epsilon$-box containing $p$.

For the ease of our analysis, we define the neighbors of a hypergrid.
\begin{definition}\label{defn:kneighbor}
  Let $\epsilon\in \mathbb{R}$. Let $g_1, g_2$ be two $\eps$-boxes in a hypergrid where $\id(g_1) = (x_1,\dots,x_d)$ and $\id(g_2) = (x'_1,\dots,x'_d)$. We say that $g_1$ and $g_2$ are each other's \emph{$\epsilon$-neighbor} if 
\begin{align}
    \sqrt{\sum_{i=1}^d \|x_i-x'_i\|^2} \leq \epsilon
\end{align}
\end{definition}

Note that the number of $\epsilon$-neighbors of a $\epsilon$-box is at most $(2\sqrt{d}+1)^d$. We also have the following observation:
\begin{observation}\label{obs:2step}
Let $p_1,p_2\in [0,L]^d$ be any two distinct points.  
\begin{itemize}
    \item If $p_1$ and $p_2$ are in the same $\epsilon$-box, then $\Delta(p_1, p_2) \leq \epsilon$.
    \item If $\Delta(p_1, p_2)\leq \epsilon$, then $g(\id(p_1))$ must be an $\epsilon$-neighbor of $g(\id(p_2))$. 
\end{itemize} 
\end{observation}

To solve $\eCP$ with quantum walk, we need data structures to keep track of the pairs that have distance at most $\epsilon$. The desired data structure should have size $\tilde{O}(n^{2/3})$, insertion/deletion time $O(\log n)$, and one should be able to check whether there exist pairs of distance at most $\epsilon$ in time $O(\log n)$. In addition, as pointed out in~\cite{ambainis07}, the data structure should have the following two properties:
\begin{itemize}
  \item the data structure should have the bounded \emph{worst-case} performance rather than \emph{average-case} performance;
  \item the representation of the data structure should be history-independent, i.e., the data is uniquely represented regardless of the order of insertions and deletions.
\end{itemize}

We need the first property since the data structure may take too long for some operations, and this is not acceptable. The second property is required because, otherwise, the interference of quantum states would be messed up. In~\cite{ambainis07}, a hash table and a skip list is used to for solving the element distinctness problem using quantum walks. In~\cite{BJLM13}, a simpler data structure, namely, a radix tree, is used to achieve the same performance. More details of using a radix tree to solve the element distinctness can be found in~\cite{Jeffery2014}. Similar to the quantum data structure model in~\cite{ambainis07,BJLM13,Jeffery2014}, we need the \emph{quantum random access gate} to efficiently access data from a quantum memory, whose operation is defined as:
\begin{align}
    \label{eq:qram-gate}
    \ket{i, b, z_1,\dots, z_m} \mapsto \ket{i, z_i, z_1, \ldots, z_{i-1}, b, z_{i+1}, z_m},
\end{align}
where $\ket{z_1,\dots, z_m}$ is some data in a quantum memory with $m$ qubits. We assume this operation takes $O(\log m)$ time.

In the remainder of this section, we present two quantum algorithms for solving $\eCP$. The data structures of both versions are based on the \textit{augmented radix tree}, which we discuss in detail in the following subsection.

\subsection{Radix tree for at most one solution}
\label{sec:radixtree}
The purpose of the augmented radix tree is to quickly locate the points in an $\epsilon$-box given its id. An ordinary radix tree is a binary tree that organizes a set of keys which are represented as binary strings. Each edge is labeled by a substring of a key and each leaf is labeled by a key such that concatenating all the labels on the path from the root to a leaf yields the key for this leaf. In addition, for each internal node, the labels of the two edges connecting to two children start with different bit. Note that in this definition, we implicitly merge all internal nodes that have only one child. The radix tree is uniquely represented for any set of keys. An example of a radix tree is shown as \cref{fig:radix-tree}.

\begin{figure}[t]
  \centering
  \includegraphics[width=0.3\textwidth]{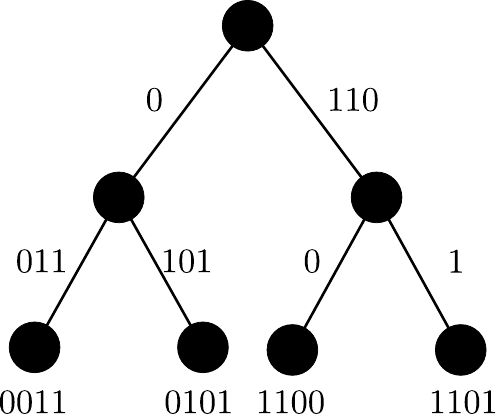}
  \caption{The uniquely represented radix tree that stores the keys $\{0011, 0101, 1100, 1101\}$.} \label{fig:radix-tree}
\end{figure}

Our basic radix tree is essentially the one in~\cite{BJLM13, Jeffery2014} with modification on the nodes' internal structure. We highlight the extra information stored in the radix tree.
First we use a \emph{local counter} to store the number of points in this $\epsilon$-box; second, we use a \emph{flag} in each leaf node to indicate whether there is a point in this $\epsilon$-box that is in some pair with distance at most $\epsilon$. The flag bit in an internal node is the $\mathrm{OR}$ of the ones in its children. The local counter in each internal node is the sum of the local counters in its children. We also store at most two points that are in the $\epsilon$-box corresponding to this node. More precisely, let $S$ be a subset of indices of the input points. We use $\tau(S)$ to denote the radix tree associated with $S$.
Then, $\tau(S)$ consists of at most $r\ceil{\log r}$ nodes. Each node consists of the following registers:
\begin{align}
  \mathcal{D} \times \mathcal{M}_1 \times \mathcal{M}_2 \times \mathcal{M}_3 \times \mathcal{C} \times \mathcal{F} \times \mathcal{P}_1 \times \mathcal{P}_2,
\end{align}
where $\mathcal{D}$ stores the id of an $\epsilon$-box for a leaf (and a substring of an id for an internal node) using $O(d\log(L/\epsilon))$ bits. $\mathcal{M}_1, \mathcal{M}_2$, and $\mathcal{M}_3$ use $O(\log n)$ bits to store the pointers to its parent, left child, and right child, respectively as well as the labels of the three edges connecting them to this node, $O(\log n)$ bits to store the labels of the three edges incident to it. $\mathcal{C}$ uses $O(\log n)$ bits to store the local counter. $\mathcal{F}$ stores the flag bit. $\mathcal{P}_1$ and $\mathcal{P}_2$ stores the coordinates of at most two points in this $\epsilon$-box, which takes $O(d\log L)$ bits. The two points are stored in ascending order of their indices.

We need to pay attention to the layout of $\tau(S)$ in memory. We use three times more bits than needed to store $\tau(S)$, this will ensure that there are always more than $1/3$ of the bits that are free. We divide the memory into cells where each cell is large enough to store one leaf node of $\tau(S)$. Besides $\tau(S)$, we also store a bitmap $\mathcal{B}$, which takes $O(\log n)$ bits to encode the current free cells (with ``1'' indicating occupied and ``0'' indicating free). To make the radix tree history-independent, we use a quantum state which is the uniform superposition of basis states $\ket{\tau(S), B}$ for all possible valid layout of $\tau(S)$ and it corresponds to the bitmap $\mathcal{B}$.

Insertion and deletion from $\tau(S)$ takes $O(\log n)$ time. Checking the presence of an $\epsilon$-close pair takes constant time --- we just need to read the flag bit in the root. Preparing the uniform superposition of all $i \in S$ can be done in $O(\log n)$ time by performing a controlled-rotation on each level of the radix tree where the angles are determined by the local counters in the two children of a node.

In the following subsections, we present the two versions of our algorithms. The first version invokes the quantum walk framework only once and its data structure maintains the existence of an $\epsilon$-close pair. The second version uses a much simpler data structure, but it is only capable of handling $\eCP$ with a unique solution. Hence it requires invoking the quantum walk framework multiple times to solve the general $\eCP$. These two quantum algorithms have almost the same time complexity. 

\subsection{Single-shot quantum walk with complicated data structure}
\label{sec:one-shot}
To handle multiple solutions, our data structure is a composition of an augmented radix tree, a hash table, and a skip list. We give a high-level overview of our data structure as follows. Recall that by the discretization of the space into $\epsilon$-boxes,  it is possible that a pair of points in different $\epsilon$-boxes have distance at most $\epsilon$, but one only needs to check $(2\sqrt{d}+1)^d$ $\epsilon$-neighbors to detect such a case. We maintain a list of points for each nonempty $\epsilon$-box in an efficient way. A hash table is used to store the tuple $(i, p_i)$ which is used to quickly find the point $p_i$, given its index $i$. The points are also stored in a skip list for each nonempty $\epsilon$-box, ordered by its index $i$, which allows for quick insertion and deletion of points. Each $\epsilon$-box is encoded into a unique key, and a radix tree is used to store such key-value pairs, where the value is associated with a skip list. The flag bits in this radix tree maintain the presence of an $\epsilon$-close pair.

In the following, we present the details of the data structure and show it has all the desired properties.

\paragraph{Hash table.}
The hash table we use is almost the same as the one used in~\cite{ambainis07}, except that we do not store the $\floor{\log r}$ counters in each bucket to facilitate the diffusion operator (which is handled easily here in the quantum walk on a Johnson graph). Our hash table has $r$ buckets, where each bucket contains $\ceil{\log n}$ entries. We use a fixed hash function $h(i) = \floor{ir/n}+1$ to hash $\{1, \ldots, n\}$ to $\{1, \ldots, r\}$. That is, for $j\in [r]$, the $j$-th bucket contains the entries for $(i, p_i)$ in ascending order of $i$, where $i \in S$ and $h(i) = j$. 

The entry for $(i, p_i)$ contains the tuple $(i, p_i)$ and $\ceil{\log n} + 1$ pointers to other entries. These pointers are used in the skip list which we will describe below. The memory size of each entry is hence $O(\log^2n + d\log L)$ and there are $O(r\log n)$ entries. Therefore, the hash table uses $O(rd\log^3(n + dL))$ qubits. 

It is possible that more than $\ceil{\log n}$ points are hashed into the same bucket. However, as shown in~\cite{ambainis07}, this probability is small.

\paragraph{Skip list.}
The skip list we use closely follows that in~\cite{ambainis07}, except that the elements $p_i$ in our skip list is ordered by its index $i$. We construct a skip list for each $\epsilon$-box containing at least one point to store the points in it. For each $i \in S$, $p_i$ belongs to exactly one skip list. 
Also, for $i \in S$, we randomly assign a level $\ell_i \in [0, \ldots, \ell_{\max}]$ where $\ell_{\max} = \ceil{\log n}$. The skip list associated with a $\epsilon$-box has $\ell_{\max}+1$ lists, where the level-$\ell$ list consists of all $i \in S$ such that $\ell_i \geq \ell$ and $p_i$ is in this $\epsilon$-box. Hence, the level-0 list consists of all $i \in S$ for $p_i$ in this $\epsilon$-box. Each element of the level-$\ell$ list has a specific pointer to the next element in this level, or to 0 if there is no next element. Each skip list contains a start entry that does not contain any $(i, p_i)$ information but $\ell_{\max}+1$ pointers to the first element of the each level. This start entry is stored in a leaf node of the augmented radix tree (which we will describe below) corresponding to this $\epsilon$-box. In each skip list, we do not allocate memory for each node. Instead, each pointer is pointing to an entry of the hash table. The pointers are stored in the hash table (for the internal entries of each level) and in the radix tree (for the start entry). An example of a skip list is shown in \cref{fig:skip-list}.

\begin{figure}[t]
  \centering
  \includegraphics[width=0.6\textwidth]{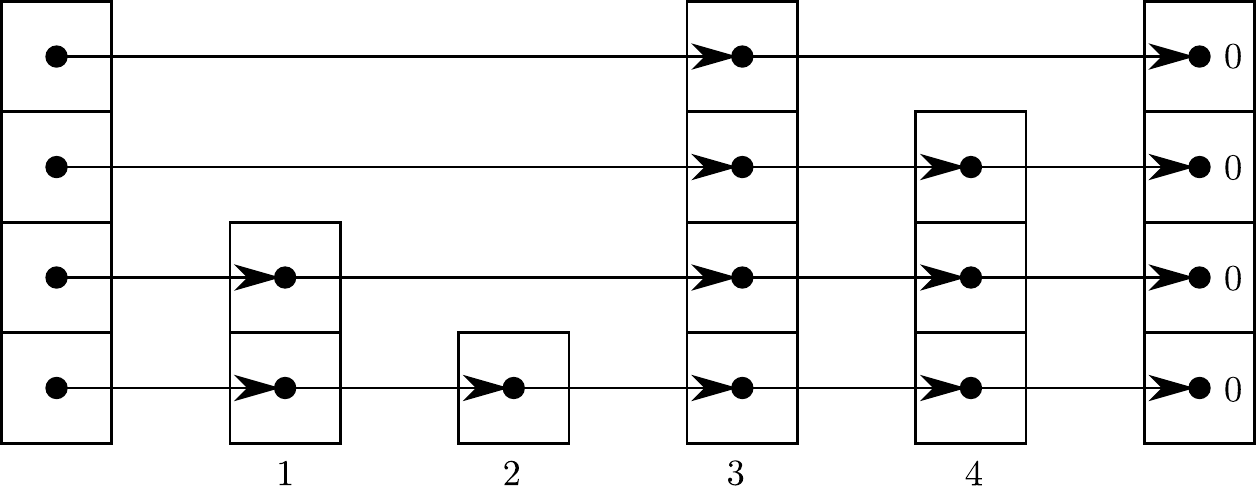}
  \caption{An example of a skip list that stores $\{1, 2, 3, 4\}$.} \label{fig:skip-list}
\end{figure}

Given $i\in S$, we can search for $p_i$ as follows. We start from the start entry of the level-$\ell_{\max}$ list and traverse each element until we find the last element $j_{\ell_{\max}}$ such that $j_{\ell_{\max}} < i$. Repeat this for levels $\ell_{\ell_{\max}-1}, \ldots, \ell_0$ and at each level start from the element that ended the previous level. At level-$0$, we obtain the element $j_0$. Then, the next element of $j_0$ is where $p_i$ should be located (if it is stored in this skip list) or be inserted.

Each $i \in S$ is randomly assigned a level $\ell_i$ at the beginning of computation that does not change during the computation. More specifically, $\ell_i = \ell$ with probability $1/2^{\ell+1}$ for $\ell < \ell_{\max}$ and with probability $1/2^{\ell_{\max}}$ for $\ell = \ell_{\max}$. This can be achieved using $\ell_{\max}$ hash functions $h_1, \ldots, h_{\ell_{\max}} : [n] \rightarrow \{0, 1\}$. In this way, each $i \in [n]$ has level $\ell<\ell_{\max}$ if $h_1(i) = \cdots = h_\ell(i) = 1$ but $h_{\ell+1}(i) = 0$; and it has level $\ell_{\max}$ if $h(i) = \cdots = h_{\ell_{\max}}(i) = 1$. In this quantum algorithm, we use an extra register to hold the state $\ket{h_1,\ldots,h_{\ell_{\max}}}$ which is initialized to a uniform superposition of all possible such functions from a $d$-wise independent family of hash functions (see~\cite[Theorem 1]{ambainis07}) for $d = \ceil{4 \log n + 1}$. During the execution of the quantum algorithm, a hash function from the hashing family is chosen depending on the state in this register.

At first glance, the skip list has the same role as the hash table -- finding $p_i$ given index $i$. 
However, they have very different purposes in our algorithm. Recall that each nonempty $\epsilon$-box is associated with a skip list, which is used to quickly insert and delete a point in this $\epsilon$-box. The number of points in this $\epsilon$-box can be as small as one and as large as $r$ (in the extreme case where all the points are in the same $\epsilon$-box).
Hence, we cannot afford to have a fixed length data structure (such as a hash table or a sorted array) to store these points. In addition, to support quick insertion and deletion, a skip list is a reasonable choice (against an ordinary list). The purpose of the hash table can be viewed as a uniquely represented memory storing all the $r$ points that can be referred to by the skip lists.

\paragraph{Augmented radix tree.}
We augment the radix tree described in \cref{sec:radixtree} to handle multiple solution. In this augmented radix tree, we do not need the registers $\mathcal{P}_1$ and $\mathcal{P}_2$. Instead, we use $\ceil{\log n}$ pointers $\mathcal{L}_1, \ldots, \mathcal{L}_{\ceil{\log n}}$ as the start entry of a skip list. These pointers uses $O(\log^2 n)$ bits. In addition, we use an \emph{external counter} in the leaf nodes to record whether there is a point in other $\epsilon$-boxes that is at most $\epsilon$-away from a point in this $\epsilon$-box, which uses $O(\log n)$ bits. More formally, let $\tau'(S)$ be the augmented radix tree associated with $S$. 
Each node of $\tau'(S)$ consists of the following registers
\begin{align}
  \mathcal{D} \times \mathcal{M}_1 \times \mathcal{M}_2 \times \mathcal{M}_3 \times \mathcal{E} \times \mathcal{C} \times \mathcal{F} \times \mathcal{E} \times \mathcal{L}_1 \times \cdots \times \mathcal{L}_{\ceil{\log n}}.
\end{align}

Next, we present how to perform the required operations on $S$ with our data structure.

\paragraph{Checking for $\epsilon$-close pairs.}To check the existence of an $\epsilon$-close pair, we just read the flag in the root of the radix tree. If the flag is set, there is at most one $\epsilon$-close pair in $S$, and no such pairs otherwise. This operation takes $O(1)$ time.

\paragraph{Insertion.} 
Given $(i, p_i)$, we perform the insertion with the following steps:
\begin{enumerate}
  \item Insert this tuple into the hash table.
  \item Compute the id, $\id(p_i)$, of the $\epsilon$-box which $p_i$ belongs to. Denote this $\epsilon$-box by $g(\id(p_i))$.
  \item Using $\id(p_i)$ as the key, check if this key is already in $\tau'(S)$, if so, insert $i$ into the skip list corresponding to $g(\id(p_i))$; otherwise, first create a uniform superposition of the addresses of all free cells into another register, then create a new tree node in the cell determined by this address register and insert it into the tree. The pointers for the start entry of the skip list is initially set to 0. Insert $i$ into this skip list. Let $\tau'(S, g(\id(p_i))$ denote the leaf node in $\tau'(S)$ corresponding to $g(\id(p_i))$.
  \item Increase the local counter $\mathcal{C}$ in $\tau'(S, g(\id(p_i)))$ by 1.
  \item Use \cref{proc:update_insertion} to update the external counters $\mathcal{E}$ and flags $\mathcal{F}$ in $\tau'(S, g(\id(p_i)))$ as well as in the leaf nodes corresponding to the neighbor $\epsilon$-boxes of $g(\id(p_i))$.
\end{enumerate}

Note that the first step takes at most $O(\log n)$ time. The second step can be done in $O(d)$ time. In \cref{line:neighbors,line:neighbors2}, the number of $\epsilon$-neighbors to check is at most $(2\sqrt{d}+1)^d$.

\begin{procedure}
  \SetKwInOut{Input}{input}\SetKwInOut{Output}{output}
  \Input{$(i, p_i)$, The leaf node in $\tau'(S)$ corresponding to the $\epsilon$-box $g(\id(p_i))$, denoted by, $\tau'(S, g(\id(p_i)))$.}
  \uIf{the local counter $\mathcal{C}=1$ in $\tau'(S, \id(p_i))$}{
    \For{all $\epsilon$-box $g'$ that is a $\epsilon$-neighbor (see \cref{defn:kneighbor}) of $g(\id(p_i))$ where the local counter $\mathcal{C}=1$ in $\tau'(S, g')$ and the distance between $p_i$ and the point in $g'$ is at most $\epsilon$ \label{line:neighbors}} {
      Increase the external counter $\mathcal{E}$ of $\tau'(S, g')$ by 1\;
      Increase the external counter $\mathcal{E}$ of $\tau'(S, g(\id(p_i)))$ by 1\;
      \If{the external counter $\mathcal{E}$ in $\tau'(S, g')$ was increased from 0 to 1} {
        Set the flag $\mathcal{F}$ in $\tau'(S, g')$ \;
        Update the flag $\mathcal{F}$ in the nodes along the path from $\tau'(S, g')$ to the root of $\tau'(S)$ \;
      }
    }
    \If{the external counter $\mathcal{E} > 1$ in $\tau'(S, g(\id(p_i)))$ } {
      Set the flag $\mathcal{F}$ in $\tau'(S, g(\id(p_i)))$ \;
      Update the flag $\mathcal{F}$ in the nodes along the path from $\tau'(S, \id(p_i))$ to the root of $\tau'(S)$ \;
    }
  }
  \ElseIf{the local counter $\mathcal{C}=2$ in $\tau'(S, \id(p_i))$}{
    Set the flag $\mathcal{F}$ in $\tau'(S, g(\id(p_i)))$ \;
    Update the flag $\mathcal{F}$ in the nodes along the path from $\tau'(S, g(\id(p_i)))$ to the root of $\tau'(S)$ \;
    Set the external counter $\mathcal{E}=0$ in $\tau'(S, \id(p_i))$ \;
    Let $i'$ be the other index (than $i$) stored in the skip list corresponding to $g(\id(p_i))$ 
    \;
    \For{all $\epsilon$-box $g'$ that is a $\epsilon$-neighbor of $g(\id(p_i))$ where the local counter $\mathcal{C}=1$ in $\tau'(S, g')$ and the distance between $p_{i'}$ and the point in $g'$ is at most $\epsilon$ \label{line:neighbors2}} {
      Decrease the external counter of $\tau'(S, g')$ by 1\;
      \If{the external counter $\mathcal{E}$ in $\tau'(S, g')$ was decreased from 1 to 0} {
        Unset the flag $\mathcal{F}$ in $\tau'(S, g')$ \;
        Update the flag $\mathcal{F}$ in the nodes along the path from $\tau'(S, g')$ to the root of $\tau'(S)$ \;
      }
    }
  }
  \caption{Updating nodes for insertion.()}
  \label[procedure]{proc:update_insertion}
\end{procedure}

To obtain a uniform superposition of the addresses of all free cells, we first create a uniform superposition of all possible addresses to access to the bitmap $\ket{\mathcal{B}}$. We also use an auxiliary register that is initialized to $\ket{0}$. Then, the quantum random access gate defined in \cref{eq:qram-gate} is applied on the register holding the uniform superposition of all addresses, the auxiliary register, and the bitmap register, which is effectively a SWAP operation on the second register and the corresponding bit in $\ket{\mathcal{B}}$. The auxiliary register remains $\ket{0}$ if and only if the address in the first register is free. Since the size of memory space is chosen so that the probability of seeing a free cell is at least $1/3$ (and we know exactly this probability based on how many cells have already been used), we can add an extra register and apply a one-qubit rotation to make the amplitude of the second register being $\ket{0}$ exactly $1/2$. Hence, using \emph{one} iteration of the oblivious amplitude amplification (which is a generalized version of Grover's search algorithm. See~\cite{BCCKS2017} and~\cite{MW2005}) with the second register being the indicator, we obtain the uniform superposition of the addresses of all free cells. This cost if $O(\log n)$.

In~\cite{ambainis07}, it was shown that with high probability, insertion into the skip list can be done in $O(d+\log^4(n+L))$ time. Hence, with high probability, the insertion costs $O(d+\log^4(n+L) + d(2\sqrt{d}+1)^d)$ time, where $O(d(2\sqrt{d}+1)^d)$ is the time for checking neighbors. 
To further reduce the running time, we can just stop the skip list's insertion and deletion procedures after $O(d+\log^4(n+L))$ time, which only causes little error (see \cref{lemma:dt-error}).

\paragraph{Deletion.}
Given $(i, p_i)$, we perform the following steps to delete this tuple from our data structure.
\begin{enumerate}
  \item Compute the id, $\id(p_i)$, of the $\epsilon$-box which $p_i$ belongs to, and denote this $\epsilon$-box by $g(\id(p_i))$.
  \item Using $\id(p_i)$ as the key, we find the leaf node in $\tau'(S)$ that is corresponding to $g(\id(p_i))$.
  \item Remove $i$ from the skip list, and decrease the local counter $\mathcal{C}$ in $\tau'(S, g(\id(p_i)))$ by 1.
  \item Use \cref{proc:update_deletion} to update the external counters $\mathcal{E}$ and flags $\mathcal{F}$ in $\tau'(S, g(\id(p_i)))$ as well as in leaf nodes corresponding to the neighbor $\epsilon$-boxes of $g(\id(p_i))$.
  \item If the local counter $\mathcal{C}=0$ in this leaf node, remove $\tau'(S, g(\id(p_i)))$ from $\tau'(S)$, and update the bitmap $\mathcal{B}$ in $\tau'(S)$ that keeps track of all free memory cells. 
  \item Remove $(i, p_i)$ from the hash table.
\end{enumerate}

Note that the first step can be done in $O(d)$ time. The second step can be done in $O(\log n)$ time. \cref{proc:update_deletion} has the same time complexity with \cref{proc:update_insertion}. Hence, the cost for the deletion procedure is the same as that for insertion.

\begin{procedure}
  \SetKwInOut{Input}{input}\SetKwInOut{Output}{output}
  \Input{$(i, p_i)$, The leaf node in $\tau'(S)$ corresponding to the $\epsilon$-box $g(\id(p_i))$, denoted by, $\tau'(S, g(\id(p_i)))$.}
  \uIf{the local counter $\mathcal{C}=0$ in $\tau'(S, \id(p_i))$}{
    Unset the flag $\mathcal{F}$ in $\tau'(S, g(\id(p_i)))$ \;
    Update the flag $\mathcal{F}$ in the nodes along the path from $\tau'(S, \id(p_i))$ to the root of $\tau'(S)$ \;
    Set the external counter $\mathcal{E} = 0$ in $\tau'(S, \id(p_i))$ \;
    \For{all $\epsilon$-box $g'$ that is a $\epsilon$-neighbor (see \cref{defn:kneighbor}) of $g(\id(p_i))$ where the local counter $\mathcal{C}=1$ in $\tau'(S, g')$ and the distance between $p_i$ and the point in $g'$ is at most $\epsilon$ \label{line:neighbors3}} {
      Decrease the external counter $\mathcal{E}$ of $\tau'(S, g')$ by 1\;
      \If{the external counter $\mathcal{E}$ in $\tau'(S, g')$ was decreased from 1 to 0} {
        Unset the flag $\mathcal{F}$ in $\tau'(S, g')$ \;
        Update the flag $\mathcal{F}$ in the nodes along the path from $\tau'(S, g')$ to the root of $\tau'(S)$ \;
      }
    }
  }
  \ElseIf{the local counter $\mathcal{C}=1$ in $\tau'(S, \id(p_i))$}{
    Let $i'$ be the only index stored in the skip list corresponding to $g(\id(p_i))$ \;
    \For{all $\epsilon$-box $g'$ that is a $\epsilon$-neighbor of $g(\id(p_i))$ where the local counter $\mathcal{C}=1$ in $\tau'(S, g')$ and the distance between $p_{i'}$ and the point in $g'$ is at most $\epsilon$ \label{line:neighbors4}} {
      Increase the external counter $\mathcal{E}$ of $\tau'(S, g')$ by 1\;
      Increase the external counter $\mathcal{E}$ of $\tau'(S, g(\id(p_i)))$ by 1\;
      \If{the external counter $\mathcal{E}$ in $\tau'(S, g')$ was increased from 0 to 1} {
        Set the flag $\mathcal{F}$ in $\tau'(S, g')$ \;
        Update the flag $\mathcal{F}$ in the nodes along the path from $\tau'(S, g')$ to the root of $\tau'(S)$ \;
      }
    }
    \If{the external counter $\mathcal{E}=0$ in $\tau'(S, g(\id(p_i)))$} {
      Unset the flag $\mathcal{F}$ in $\tau'(S, g(\id(p_i)))$ \;
      Update the flag $\mathcal{F}$ in the nodes along the path from $\tau'(S, \id(p_i))$ to the root of $\tau'(S)$ \;
    }
  }
  \caption{Updating nodes for deletion.()}
  \label[procedure]{proc:update_deletion}
\end{procedure}

\paragraph{Finding a $\epsilon$-close pair.} We just read the flag in the root of the radix tree and then go to a leaf whose flag is $1$. Check the local counter $\mathcal{C}$ of the node. if it is at least $2$, output the first two elements in skip list. Otherwise, we find the $\epsilon$-neighbor of the current node whose $C=1$ and then output the points in that $\epsilon$-neighbor and the current node.

\paragraph{Uniqueness.}
The uniqueness of our data structure follows from the analysis of~\cite{ambainis07,BJLM13,Jeffery2014}. More specifically, the hash table is always stored in the same way, as each $i \in S$ is stored in the same bucket for the fixed hash function and in each bucket, elements are stored in ascending order of $i$. The skip list is uniquely stored once the hash functions $h_1, \ldots, h_{\ell_{\max}}$ is determined. The shape of the radix tree is unique for $S$, but each node can be stored in different locations in memory. We use a uniform superposition of all possible memory organizations (by keeping track of the bitmap for free cells) to keep the quantum state uniquely determined by $S$.

\paragraph{Correctness.}
In the following, we argue that our data structure has the desired properties. First, we prove the correctness.

\begin{lemma}
  The flag bit in the root of $\tau'(S)$ is set if and only if there exist  distinct $i, j \in S$ such that $|p_i - p_j| \leq \epsilon$.
\end{lemma}
\begin{proof}
  We show that after each insertion and deletion, the data structure maintains the following conditions, and then lemma follows.
  \begin{enumerate}
    \item The flag bit of each leaf node of $\tau'(S)$ is set if and only if either its local counter is at least 2, or its external counter is at least 1.
    \item The external counter of a leaf node $\tau'(S, g')$ is nonzero if and only if $g'$ contains only one point $p$, and there exists another $p'$ in another $\epsilon$-box $g''$ such that $|p - p'| \leq \epsilon$.
  \end{enumerate}
  It is easy to check that the first condition is maintained for each insertion and deletion. We show the second condition holds during insertions and deletion. For insertions, we consider the first case where a point $p$ is just inserted into the $\epsilon$-grid $g'$, and $p$ is its only point. The first for-loop in \cref{proc:update_insertion} updates other $\epsilon$-boxes that have only one point to maintain the second condition. We consider the second case where $g'$ already contains $p'$ and $p$ is just inserted, then the external counter in $g'$ should be 0, and the second for-loop in \cref{proc:update_insertion} updates other $\epsilon$-boxes that have only one point using the information of $p'$. This maintains the second condition. For deletions, there are also two cases. First, consider $p$ is the only one point in $g'$ and it is just deleted. We use the first for-loop in \cref{proc:update_deletion} to update the $\epsilon$-boxes that has only one point using the information of $p$ to maintain the second condition. Second, there is another point $p'$ left in $g'$ after deleting $p$. In this case, we start to check the external counter in $g'$. We use the second for-loop in \cref{proc:update_deletion} to check other $\epsilon$-boxes that have only one point using the information of $p'$ and update the corresponding external counter to maintain the second condition.
\end{proof}

\paragraph{Imperfection of the data structures and error analysis.}
Our data structure is not perfect. As indicated by Ambainis~\cite{ambainis07}, there are two possibilities that it will fail. First, the hash table might overflow. Second, it might take more that $\ceil{\log n}$ time to search in a skip list. To resolve the first problem, we fix the number of entries in each bucket to be $\ceil{\log n}$ and treat any overflow as a failure. For the second problem, we stop the subroutine for accessing the skip list after $O(\log n)$ steps, and it causes an error in some cases. The original error analysis can be directly used in our case, as our hash table doesn't change the structure or the hash function, and our skip lists can be viewed as breaking the skip list in~\cite{ambainis07} into several pieces (one for each nonempty $\epsilon$-box), and each insertion/deletion only involves one of them. Hence, the discussion in~\cite[Section 6]{ambainis07} can be directly adapted to our case: 
\begin{lemma}[Adapted from \cite{ambainis07}]
  \label{lemma:dt-error}
  Let $\ket{\psi}$ be the final state of our algorithm (with imperfect data structures) and $\ket{\psi'}$ be the final state with the perfect data structure. Then $\|\ket{\psi} - \ket{\psi'}\| \leq O(1/\sqrt{n})$.
\end{lemma}
\begin{proof}[Sketch of proof]
There are two places where the data structure may give error: first, the hash table may have overflow, and second, the algorithm cannot find the required element in the skip lists in the desired time. The original proof showed that the probability that any of these errors happens is negligible, and thus the two-norm distance between $\ket{\psi}$ and $\ket{\psi'}$ can be bounded. Here, our data structure combining hash table, skip list, and radix tree, only has errors from hash tables and skip lists. The radix tree which has no error can be viewed as applying additional unitaries on $\ket{\psi}$ and $\ket{\psi'}$, and this does not change the distance between the two states. Since the probability that the errors from hash tables and skip lists happen are negligible by following the same analysis in~\cite{ambainis07}, we can conclude that the two-norm distance between $\ket{\psi}$ and $\ket{\psi'}$ is small.  
\end{proof}

\paragraph{Time complexity for $\eCP$.}
We use the quantum walk framework reviewed in \cref{sec:qw-framework} to solve $\eCP$. We first build the Johnson graph for $\eCP$, which is similar to that for \textsf{ED} in \cref{sec:qw-framework}. The vertices of the Johnson graph are $S \subseteq [n]$ with $|S| = n^{2/3}$ and $S$ is marked if there exist distinct $i, j \in S$ such that $\Delta(p_i, p_j) \leq \epsilon$. 
We use $\ket{S, d(S)}$ to represent the quantum state corresponding to $S$, where $d(S)$ is the data structure of $S$ defined in \cref{sec:radixtree}. 
Consider the three operations:
\begin{itemize}
  \item \textbf{Steup:} with cost $\mathsf{S}$,
  preparing the initial state 
    \begin{align}
      \ket{\pi} = \frac{1}{\sqrt{\binom{n}{n^{2/3}}}}\sum_{S \subseteq [n]: |S|=n^{2/3}}\sqrt{\pi_S}\ket{S, d(S)}.
    \end{align}
  \item \textbf{Update:}  with cost $\mathsf{U}$, applying the transformation 
    \begin{align}
      \ket{S, d(S)}\ket{0} \mapsto \ket{S, d(S)}\sum_{S' \subseteq[n]: |S\cap S'| = n^{2/3}-1}\sqrt{p_{SS'}}\ket{S', d(S')},
    \end{align}
    where $p_{SS'} = \frac{1}{n^{2/3}(n-n^{2/3})}$.
  \item \textbf{Checking:} with cost $\mathsf{C}$, applying the transformation: 
    \begin{align}
      \ket{S, d(S)} \mapsto \begin{cases} -\ket{S, d(S)} \quad &\text{if $S\in M$} \\ \ket{S, d(S)} \quad &\text{otherwise,}\end{cases}
    \end{align}
    where $M$ is the set of marked states.
\end{itemize}

We have the following result.
\begin{theorem}
  \label{thm:ecp1}
  There exists a quantum algorithm that with high probability solves $\eCP$ with time complexity $O(n^{2/3}(d+\log^4(n+L) + d(2\sqrt{d}+1)^d))$.
\end{theorem}

\begin{proof}
  The Johnson graph has $\varepsilon \geq 1/n^{2/3}$ and the Markov chain has spectral gap $\delta \geq 1/n^{2/3}$. For the setup operation, $\mathsf{S} = O(n^{2/3}(d+\log^4(n+L) + d(2\sqrt{d}+1)^d))$, since preparing the uniform superposition for all $\ket{S}$ costs $O(\log n)$ Hadamard gates and we need to do $n^{2/3}$ insertions to prepare the data structure. Each insertion costs $O(d+\log^4(n+L) + d(2\sqrt{d}+1)^d)$. For the update operation, we can implement the quantum walk operator as described in~\cite{Jeffery2014}: we use $\ket{S,d(S)}\ket{i,j}$ to represent $\ket{S,d(S)}\ket{S',d(S')}$ where $S'$ is obtained from $S$ by adding index $i$ and removing index $j$. Then the diffusion can be implemented by preparing a uniform superposition of all $i \in S$ and a uniform superposition of all $j \not\in S$, which takes time $O(\log n)$, and the ``SWAP'' operation can be implemented by a unitary that maps $\ket{S, d(S)}\ket{i,j}$ to $\ket{S',d(S')}\ket{j,i}$. In this way, the update operation uses $O(1)$ insertion and deletion to construct $d(S')$ from $d(S)$, and hence $\mathsf{U} = O(d+\log^4(n+L) + d(2\sqrt{d}+1)^d)$. The checking operation can be done in $O(1)$ time with the data structure. Therefore, by \cref{lem:framework}, the time complexity is $O(\mathsf{S} + \frac{1}{\sqrt{\varepsilon}}(\frac{1}{\sqrt{\delta}}\mathsf{U} + \mathsf{C})) = O(n^{2/3}(d+\log^4(n+L) + d(2\sqrt{d}+1)^d))$.
\end{proof}

By \cref{lemma:cp-ecp}, we have the following corollary.
\begin{corollary}
  There exists a quantum algorithm that with high probability solves \textsf{CP} with time complexity $O(n^{2/3}\cdot (d+\log^4(n+L) + d(2\sqrt{d}+1)^d)\cdot (m+\log d))$.
\end{corollary}

\begin{remark}
For $d=O(1)$ dimension and $m=O(\log n)$ digits of each coordinate of the points, the running time of the single-shot quantum algorithm is $O(n^{2/3}\cdot \log^5 n)$.
\end{remark}

\subsection{Multiple-trial quantum walks with simple data structure}
\label{sec:multi-shots}

In the previous subsection, we provide a quantum algorithm which solves $\eCP$ in time $O(n^{2/3}(d+\log^4 n + d(2\sqrt{d}+1)^d))$, where the logarithmic cost is mainly from the cost of the skip list. In this section we present a quantum algorithm which only requires the radix tree, and thus improve the running time. 
The caveat is that, with only the radix tree data structure, the insertion would fail if there are more than one $\epsilon$-close pairs. As a result, we need to keep shrinking the size of the problem using~\cite[Algorithm 3]{ambainis07} until there is a unique solution with high probability, and then run the $\tilde{O}(n^{2/3})$ quantum algorithm for this unique-solution case. 

In the following, we first show how to solve the unique-solution $\eCP$, and then show the reduction from the multiple-solution case to the unique-solution case. 
\begin{lemma}
  \label{lem:ecp-single-solution}
  There exists a quantum algorithm that with high probability solves the unique-solution $\eCP$ with time complexity $O(n^{2/3}(\log n + d(2\sqrt{d}+1)^d))$.
\end{lemma}

\paragraph{Data structure for unique-solution.}

We use the radix tree $\tau(S)$ for $S$ defined in \cref{sec:radixtree}. In the following, we describe the necessary operations on $\tau(S)$.

\paragraph{Checking for $\epsilon$-close pair.} To check the existence of an $\epsilon$-close pair, we just read the flag bit in the root of $\tau(S)$, which takes $O(1)$ time.

\paragraph{Insertion.} Given $(i, p_i)$, we perform the following steps for insertion. First compute the id, $\id(p_i)$, of the $\epsilon$-box which $p_i$ belongs to. Denote this $\epsilon$-box by $g(\id(p_i))$. Using $\id(p_i)$ as the key, check if this key is already in $\tau(S)$. There are two cases:
\begin{itemize}
  \item $\id(p_i)$ is already in $\tau(S)$: insert $p_i$ into $\tau(S, g(\id(p_i)))$, increase the local counter in $\tau(S, g(\id(p_i)))$ by 1 and also set the flag. Then update the flag and local counter of the nodes along the path from $\tau(S, g(\id(p_i)))$ to the root. 
  \item $\id(p_i)$ is not in $\tau(S)$: create a new leaf node for $\id(p_i)$ and insert it into $\tau(S)$. Insert $p_i$ into this new leaf node, and increase the local counter in $\tau(S, g(\id(p_i)))$ by 1. Then, check the $\epsilon$-neighbors $g'$ of $\tau(S, g(\id(p_i)))$ that contains only one point $p'$ and set both flags if $p_i$ is $\epsilon$-close to $p'$, and in this case, update the flag bit and local counter on the nodes along the paths from $\tau(S, g(\id(p_i)))$ and $g'$.
\end{itemize}

\paragraph{Deletion.} Given $(i, p_i)$, we first compute the id, $\id(p_i)$ of the $\epsilon$-box that $p_i$ belongs to, and locate the corresponding leaf node $\tau(S, g(\id(p_i)))$. Decrease the local counter in $\tau(S, g(\id(p_i)))$ by 1 and update the local counter in the nodes along the path from $\tau(S, g(\id(p_i)))$ to the root. Check the number of points stored in $\tau(S, g(\id(p_i)))$. There are two possibilities:
\begin{itemize}
  \item There are two points in $\tau(S, g(\id(p_i)))$: unset the flag in $\tau(S, g(\id(p_i)))$ and update the flag bit in the nodes along the path to the root and delete $p_i$ from $\tau(S, g(\id(p_i)))$.
  \item $p_i$ is the only point in $\tau(S, g(\id(p_i)))$: check the $\epsilon$-neighbors $g'$ of $\tau(S, g(\id(p_i)))$ that contains only one point $p'$ and unset both flags if $p_i$ is $\epsilon$-close to $p'$, and in this case, update the flag bit on the nodes along the path from $\tau(S, g(\id(p_i)))$ and $g'$ to the root. Delete $p_i$ from $\tau(S, g(\id(p_i)))$ and delete $\tau(S, g(\id(p_i)))$ from $\tau(S)$.
\end{itemize}

As in \cref{sec:one-shot}, we use a bitmap register $\ket{\mathcal{B}}$ to keep track of the free cells in $\tau(S)$. For insertion, we maintain a uniform superposition of all possible free cells to insert a new radix tree node. For deletion, we update the bitmap $\ket{\mathcal{B}}$ accordingly. This ensures the uniqueness of the quantum data structure.

The correctness of the data structure is straightforward, and the time complexity is $O(\log n + d(2\sqrt{d}+1)^d)$ for both insertion and deletion. Also, preparing a uniform superposition for all $i \in S$ costs $O(\log n)$ using the local counter in each node. By a similar analysis of \cref{thm:ecp1}, we prove \cref{lem:ecp-single-solution} as follows.

\begin{proof}[Proof of \cref{lem:ecp-single-solution}]
The algorithm uses the framework in \cref{lem:framework} with the data structure we just described in this subsection, where $\mathsf{U} = O(\log n + d(2\sqrt{d}+1)^d))$,  $\mathsf{C} = O(1)$ and $\mathsf{S} = O(n^{2/3}(\log n + d(2\sqrt{d}+1)^d))$. 
Therefore, the running time of the algorithm is as claimed.
\end{proof}

Next, we show how to reduce multiple-solution $\eCP$ to unique-solution $\eCP$. A high-level overview of Ambainis's reduction in~\cite{ambainis07} is the following. We run the algorithm for unique-solution $\eCP$ several times on some random subsets of the given input. If the given subset contains solutions, then with constant probability there exists a subset which contains exactly one solution. 

\begin{definition}[\cite{ambainis07,itt05}]
Let $\mathcal{F}$ be a family of permutations on $f:[n]\rightarrow [n]$. $\mathcal{F}$ is $\epsilon$-approximate $d$-wise independent if for any $x_1,\dots,x_d\in [n]$ and for all $y_1,\dots,y_d\in [n]$, 
\begin{align}
    &\frac{1-\epsilon}{n\cdot(n-1)\cdot (n-d+1)}\leq ~ \Pr\left[\bigwedge_{i=1}^n f_i(x_i) = y_i\right] \leq ~  \frac{1+\epsilon}{n\cdot(n-1)\cdot (n-d+1)}.
\end{align}
\end{definition}

\begin{lemma}[\cite{ambainis07,itt05}]\label{lem:ehash}
Let $n$ be an even power of a prime number. For any $t\leq n$, $\epsilon>0$, there exists an $\epsilon$-approximate $t$-wise independent family $\mathcal{F}=\{\pi_j| j\in [R]\}$ of permutations $\pi_j: [n]\rightarrow [n]$ such that: 
\begin{itemize}
    \item $R = O\left( \left(n^{t^2}\cdot \epsilon^{-t}\right)^{3+o(1)}   \right)$; 
    \item given $i,j$, $\pi_j(i)$ can be computed in time $O(t\log^2 n)$.
\end{itemize}
\end{lemma}

The multiple-solution algorithm from~\cite{ambainis07} is as follows: 
\begin{algorithm}[h]
    \SetKwInOut{Input}{input}\SetKwInOut{Output}{output}
    \Input{Let $(S,\epsilon)$ be the input, and $|S| = n$.} 
    Let $T_1 = S$ and $j=1$\;
    \While {$|T_j|> n^{2/3}$}{
      Run the algorithm described in \cref{lem:ecp-single-solution} on $T_j$, and Measure the final state. If there is a pair with distance less than $\epsilon$, output the pair and stop\label{line:singlesolu} \;
      Let $q_j$ be an even power of a prime with $|T_j|\leq q_j\leq (1+\frac{1}{8})|T_j|$. Select a random permutation $\pi_j$ on $[ q_j ]$ from the $\frac{1}{n}$-approximately $4\log n$-wise independent family of permutations as in \cref{lem:ehash} \label{line:perm}\;
      Let 
            \begin{align}\label{eq:t_j}
                &T_{j+1} := \left\{ \pi_1^{-1}\cdot \pi_2^{-1} \cdots \pi_{j}^{-1}(i), \quad i\in \left[ \left\lceil \frac{4 q_j}{5} \right\rceil \right] \right\}.
            \end{align}
      $j \leftarrow j+1$ \label{line:t_j_1}\;
      }
      If $|T_j|\leq n^{\frac{2}{3}}$, then run Grover's search algorithm on $T_j$ for a pair with distance at most $\epsilon$ \label{line:basecase}\;
    \caption{The algorithm for multiple $\epsilon$-close pair}
    \label{alg:m_to_1}
\end{algorithm}

We have the main result of this subsection:
\begin{theorem}
  \label{thm:ecp2}
  There exists a quantum algorithm that with high probability solves $\eCP$ with time complexity $O(n^{2/3}\cdot (\log n + d(2\sqrt{d}+1)^d)\cdot \log^3 n)=~O(n^{2/3}\cdot \log^4 n)$ for $d=O(1)$.
\end{theorem}
\begin{proof}
We prove the running time of the algorithm here. For the correctness, one can check~\cite{ambainis07} for the detail.

By~\cref{eq:t_j}, the size of $T_{j+1}$ will be at most 
\begin{align}
    &\frac{4}{5} \cdot (1+\frac{1}{8}) |T_j| =~ \frac{9}{10}|T_j|. 
\end{align}
Therefore, the while-loop takes at most $O(\log n)$ iterations in the worst case. Let $n_j=|T_j|$ be the size of the instance in the $j$-th iteration. Then, the unique-solution algorithm in \cref{line:singlesolu} runs in $O(n_j^{2/3}\cdot (\log n_j + d(2\sqrt{d}+1)^d))$-time (\cref{lem:ecp-single-solution}), given an $O(1)$-time access to the set $T_j$. However, in~\cref{line:perm} each element of the random permutation can be computed in time $O(\log^3 n)$ according to \cref{lem:ehash} with $t = 4\log n$, which means the unique-solution algorithm will take $O(\log^3 n)$ time for each query to $T_j$. Note that we will not actually compute the whole set $T_{j+1}$, as shown in~\cref{line:t_j_1}, which takes too much time. 
Hence, the running time for the $j$-th iteration is $O(n_j^{2/3}\cdot (\log n_j + d(2\sqrt{d}+1)^d)\cdot \log^3 n)$. And the total running time for the while-loop is
\begin{align}
    &~\sum_{j=1}^{O(\log n)} O(n_j^{2/3}\cdot (\log n_j + d(2\sqrt{d}+1)^d)\cdot \log^3 n)\\
    \leq &~ O(n^{2/3}\cdot (\log n + d(2\sqrt{d}+1)^d)\cdot \log^3 n)\cdot \sum_{j=0}^{O(\log n)} \left(\frac{9}{10}\right)^{2j/3}\\
    = &~ O(n^{2/3}\cdot (\log n + d(2\sqrt{d}+1)^d)\cdot \log^3 n),
\end{align}
where the first inequality follows from $n_j\leq (\frac{9}{10})^{j-1}\cdot n$.
Finally, \cref{line:basecase} runs in time $O(n^{2/3}\log n)$. This completes the proof of the running time.  
\end{proof}

To conclude the quantum algorithms for solving $\CP$ in constant dimension, we have the following corollary that is a direct consequence of either \cref{thm:ecp1} or \cref{thm:ecp2}.
\begin{corollary}\label{cor:main_cp}
  For any $d=O(1)$, there exists a quantum algorithm that, with high probability, solves $\CP_{n,d}$ in time $\wt{O}(n^{2/3})$.
\end{corollary}

\subsection{Quantum lower bound for \textsf{CP} in constant dimensions}
\label{sec:cp-lo}
We can easily get an $\Omega(n^{2/3})$ lower bound for the quantum time complexity of \textsf{CP} in constant dimension by reducing the element distinctness problem (\textsf{ED}) to \textsf{CP}. 

\begin{theorem}[Folklore]\label{thm:cp_constant}
The quantum time complexity of $\CP$ is  $\Omega(n^{2/3})$.
\end{theorem}
\begin{proof}
We reduce \textsf{ED} to one dimensional \textsf{CP} by mapping the point $i$ with value $f(i)$ in  \textsf{ED}  the point $f(i) \in \mathbb{R}$ in \textsf{CP}. If the closest pair has distance zero, we know there is a collision $f(i)=f(j)$. If the closest pair has distance greater or equal to one, we know there is no collision. Therefore, \textsf{ED}'s  $\Omega(n^{2/3})$ query lower bound by~\cite{AS04} translates into  $\Omega(n^{2/3})$ time lower bound  for \textsf{CP}.
\end{proof}

\section{Bichromatic closest pair in constant dimensions}

Classically, bichromatic closest pair problem is harder than the closest pair problem. In constant dimension, the best algorithms for the closest pair problem are ``nearly linear'', while the algorithm by~\cite{AES91} for bichromatic closest pair problem  is ``barely subquadratic'', running in $O(n^{2-1/\Theta(d)})$-time. In quantum, we found that $\BCP$ is still harder than $\CP$ in constant dimension. In particular, we cannot adapt the quantum algorithm in previous section for solving $\BCP$ because the data structure cannot distinguish the points from two sets efficiently. We can only get a sub-linear time quantum algorithm for $\BCP$ using different approach, which is a quadratic speed-up for the classical algorithm. 

Nevertheless, we show that we can find an approximate solution for $\BCP$ with multiplicative error $1+\xi$ with quantum time complexity $\tilde{\Theta}(n^{2/3})$. The following theorem is a direct consequence of \cref{thm:bcp-up,thm:bcp-lo}.
\begin{theorem}
For any fixed dimension and error $\xi$, there is a quantum algorithm which can find an approximate solution for $\BCP$ with multiplicative error $1+\xi$ in time $\tilde{O}(n^{2/3})$. Moreover, all quantum algorithms which can find an approximate solution for $\BCP$ with arbitrary multiplicative error requires time $\Omega(n^{2/3})$. 
\end{theorem}

Similar to solving \textsf{CP}, we reduce \textsf{BCP} to its decision version of the problem, and then apply quantum algorithms to solve the decision problem. We define the decision problem as $\eBCP$. 
\begin{definition}[$\eBCP$]
In $\eBCP$, we are given two sets $A,B$ of $n$ points $\in \mathbb R^d$ and a distance measure $\Delta$. The goal is to find a pair of points $a\in A$, $b\in B$ such that $\Delta(a,b) \leq \epsilon$ if it exists and returns \textsf{no} if no such pair exists. 
\end{definition}
To address the approximate version of \BCP, we also define the approximation version of $\eBCP$ as follows:
\begin{definition}[$\xieBCP$]
  In $\xieBCP$, we are given two sets $A,B$ of $n$ points $\in \mathbb R^d$, a distance measure $\Delta$, and $\xi$. The goal is to do the following
  \begin{enumerate}
       \item If there exists a pair of points $a\in A$, $b\in B$ such that $\Delta(a,b) \leq \epsilon$, output the pair $(a,b)$.
       \item If for all pairs of points $a\in A$, $b\in B$, $\Delta(a,b) > (1+\xi)\epsilon$,  returns \textsf{no}.
   \end{enumerate}
\end{definition}

Again, we consider $\Delta(a, b) = \norm{a-b}$ as the distance measure in this work. We show that $\BCP$ reduce to $\eBCP$ in time $O(m+\log d)$, where $m$ is the number of digits of each coordinate and $d$ is the dimension.  
\begin{lemma}
  \label{lemma:bcp-ebcp}
Given an oracle $\ora$ for $\xieBCP$, there exists an algorithm $A^{\ora}$  that runs in time and query complexity $O(m+\log d )$ solves the $\xiBCP$. 
\end{lemma}
\begin{proof}
Let $(A, B,\delta)$ be an instance of the $\xiBCP$.  We first pick an arbitrary pair $a_0 \in A,b_0\in B$ and computes $\Delta(a_0,b_0)$. Then, we set $\epsilon$ to be $\Delta(a_0,b_0)/2$ and run the oracle $\ora$ to check whether there exists a distinct pair which distance is less than $\Delta(a_0,b_0)/2$ or not. If there exists such a pair, which we denote as $(a_1,b_1)$, then we set $\epsilon=\Delta(a_1,b_1)$ and call $\ora$ to check again. If there is no such a pair, then we set $\epsilon=3\Delta(a_0,b_0)/4$ and call $\ora$. We continuously run this binary search for $m+\log d$ iterations. Finally, the algorithm outputs the bichromatic closest pair.  
\end{proof}

In the subsections, we present a quantum algorithm for solving $\xiBCP$ and a quantum algorithm for exact $\BCP$. To complement the algorithmic results, we also give quantum lower bound for $\BCP$.

\subsection{Quantum algorithm for $\xiBCP$}
\label{sec:bcp-approx}
The quantum algorithm is based on the quantum walk framework on a tensor product of Johnson graphs. To begin with, we define the Johnson graphs $J_A$ and $J_B$ for $A$ and $B$, respectively. The vertices of $J_A$, denoted by $X_A$, is defined as the set $\{S \subseteq A: |S| = n^{2/3}\}$. There is an edge connecting $S$ and $S'$ if and only if $|S \cap S'| = n^{2/3}-1$. The Markov chain $M_A$ is defined on $X_A$ with $p_{SS'} = \frac{1}{n^{2/3}(n-n^{2/3})}$ when $S$ and $S'$ are connected by an edge. The Johnson graph for $J_B$ for $B$ and its corresponding Markov chain can be defined similarly. The \emph{tensor product} $M_A \otimes M_B$ is defined as the Markov chain based on $X_A \times X_B$ defined as
\begin{align}
  X_A \times X_B := \{(S_A, S_B): S_A \in X_A, S_B \in X_B\},
\end{align}
with transition probability
\begin{align}
  p_{(S_A,S_B)(S_A',S_B')} = p_{S_AS_A'}\cdot p_{S_BS_B'}.
\end{align}
A state $(S_A, S_B)$ is marked if there exists a pair $a \in S_A$ and $b \in S_B$ such that $\Delta(a, b) \leq \epsilon$.

Now, we examine the properties of $M_A \otimes M_B$. It is easy to see that $\varepsilon = \frac{\binom{n-1}{n^{2/3}-1}^2}{\binom{n}{n^{2/3}}^2} = \frac{1}{n^{2/3}}$. Let $\delta_A$ and $\delta_B$ be the spectral gap of $M_A$ and $M_B$ respectively. As a result of~\cite[Lemma 21.17]{AB09}, $\delta \geq \min\{\delta_A, \delta_B\} = \frac{1}{n^{2/3}}$. By \cref{lem:framework}, the cost for solving $\xieBCP$ is $O(\mathsf{S} + n^{1/3}(n^{1/3}\mathsf{U} + \mathsf{C}))$, where $\mathsf{S}$, $\mathsf{U}$ and $\mathsf{C}$ are the cost of quantum operations defined in \cref{sec:one-shot}. Before describing the data structure to achieve meaningful $\mathsf{S}, \mathsf{U}$, and $\mathsf{C}$, we first introduce a finer discretization scheme. In \cref{sec:cp}, we used a hypergrid consisting of $\epsilon$-boxes. 
Here, we discretize the space $[0, L]^d$ as a hypergrid consisting of $\frac{\xi\epsilon}{2\sqrt{d}}$-boxes. The following lemma guarantees that distance between a $\frac{\xi\epsilon}{2\sqrt{d}}$-box and its $\epsilon$-neighbor is at most $(1+\xi)\epsilon$.

\begin{lemma}
  \label{lemma:represent}
  Let $g$ and $g'$ be $\frac{\xi\epsilon}{2\sqrt{d}}$-boxes. If $g$ and $g'$ are $\epsilon$-neighbors, then for all $p \in g$ and $p' \in g'$, 
    $\Delta(p, p') \leq (1+\xi)\epsilon$.
\end{lemma}
\begin{proof}
Recall the definition of the $\id$ function in \cref{eq:id_fcn}. $\id(g)$ can be treated as a point, and we can measure the distance between $\id(g)$ and other points. The lemma can be proven via the triangle inequality:
\begin{align}
  \Delta(p, p') \leq \Delta(p, \id(g)) + \Delta(\id(g), \id(g')) + \Delta(p', \id(g') 
\leq \frac{\xi\epsilon}{2} + \epsilon +  \frac{\xi\epsilon}{2} 
\leq (1+\xi)\epsilon.
\end{align}
\end{proof}

In our algorithm, we need to search for all $\epsilon$-neighbors that contain the other color to report an $\epsilon$-close pair (with an multiplicative error $\xi$). It's easy to see that the number of neighbors of a box is bounded in terms of $d$ and $\xi$:

\begin{claim}\label{cla:n_neighbor_bcp}
For each $\frac{\xi\epsilon}{2\sqrt{d}}$-box, the number of $\epsilon$-neighbors is at most $\bigl(4\sqrt{d}/\xi+1\bigr)^d$.
\end{claim}

Based on this finer discretization scheme, we use the data structure defined in \cref{sec:one-shot} but with simple modifications on the radix tree. Instead of using $\mathcal{L}_1, \ldots, \mathcal{L}_{\ceil{\log n}}$ as the start entry of the skip list, we use $\ceil{\log n}$ pointers for both sets $A$ and $B$. We also need local counters $\mathcal{C}^A$ and $\mathcal{C}^B$ for both colors. Now, each node in the radix tree has the following registers:
\begin{align}
  &\mathcal{D} \times \mathcal{M}_1 \times \mathcal{M}_2 \times \mathcal{M}_3 \times \mathcal{E}^A \times \mathcal{E}^B\times \mathcal{C}^A \times \mathcal{C}^B \times \nonumber\\
  &\mathcal{F} \times \mathcal{L}_1^A \times \cdots \times \mathcal{L}_{\ceil{\log n}}^A \times \mathcal{L}_1^B \times \cdots \times \mathcal{L}_{\ceil{\log n}}^B.
\end{align}
The points in $A$ (or $B$, respectively) is organized by the skip list for $A$ (or $B$, respectively). The insertion and deletion operations are similar to the data structure in \cref{sec:one-shot}, but in the procedure for updating the local and external counters and checking $\epsilon$-neighbors, we need to consider points of the other color. We formally describe the two procedures as follows.

\paragraph{Insertion.} 
Given a point $(i, p_i, x)$, where $x\in \{A,B\}$ denotes the color. We perform the insertion with the following steps:
\begin{enumerate}
  \item Insert this tuple into the hash table corresponding to $x$.
  \item Compute the id, $\id(p_i)$, of the $\frac{\xi\epsilon}{\sqrt{d}}$-box which $p_i$ belongs to and denote it by $g(\id(p_i))$.
  \item Using $\id(p_i)$ as the key, check if this key is already in $\tau'(S)$, if so, insert $i$ into the skip list for color $x$ corresponding to $g(\id(p_i))$; otherwise, first create a uniform superposition of the addresses of all free cells into another register, then create a new tree node in the cell determined by this address register and insert it into the tree. The pointer for the start entry of the skip list is initially set to 0. Insert $i$ into this skip list. Let $\tau'(S, g(\id(p_i))$ denote the leaf node in $\tau'(S)$ corresponding to $g(\id(p_i))$.
  \item Increase the local counter $\mathcal{C}^x$ in $\tau'(S, g(\id(p_i)))$ by 1.
  \item Use \cref{proc:update_insertion_ebcp} to update the external counters $\mathcal{E}^x, \mathcal{E}^{\bar{x}}$ (here $\bar{x}$ denotes the other color than $x$) and flags $\mathcal{F}$ in $\tau'(S, g(\id(p_i)))$, the leaf nodes which are corresponding to the $\epsilon$-neighbors of $g(\id(p_i))$, and their parent nodes.
\end{enumerate}

Note that the first step takes at most $O(\log n)$ time. The second step can be done in $O(d)$ time. In \cref{line:bneighbors,line:bneighbors2}, the number of $\epsilon$-neighbors to check is at most $(\frac{4\sqrt{d}}{\xi}+1)^d$ by \cref{cla:n_neighbor_bcp}.

\begin{procedure}
  \SetKwInOut{Input}{input}\SetKwInOut{Output}{output}
  \Input{$(i, p_i,x)$, the leaf node in $\tau'(S)$ corresponding to $g(\id(p_i))$, denoted by $\tau'(S, g(\id(p_i)))$.}
  Let $\bar{x}\in \{A,B\}$ and $\bar{x}\neq x$\;
  \uIf{$\mathcal{C}^x=1$ in $\tau'(S, \id(p_i))$ and $\mathcal{C}^{\bar{x}}=0$}{
    \For{all $\epsilon$-neighbor $g'$ (see \cref{defn:kneighbor}) of $g(\id(p_i))$ where $\mathcal{C}^{\bar{x}} \geq 1$ in $\tau'(S, g')$ \label{line:bneighbors}}
     {
      Increase $\mathcal{E}^x$ of $\tau'(S, g')$ by 1\;
      Increase $\mathcal{E}^{\bar{x}}$ of $\tau'(S, g(\id(p_i)))$ by 1\;
      \If{$\mathcal{E}^x$ in $\tau'(S, g')$ was increased from 0 to 1} {
        Set the flag $\mathcal{F}$ in $\tau'(S, g')$ \;
        Update the flags $\mathcal{F}$  in the nodes along the path from $\tau'(S, g')$ to the root of $\tau'(S)$ \;
      }
    }
    \If{$\mathcal{E}^{\bar{x}} \geq 1$ in $\tau'(S, g(\id(p_i)))$} {
      Set the flag $\mathcal{F}$ in $\tau'(S, g(\id(p_i)))$ \;
      Update the flags $\mathcal{F}$ in the nodes along the path from $\tau'(S, \id(p_i))$ to the root of $\tau'(S)$ \;
    }
  }
  \ElseIf{$\mathcal{C}^x=1$ and $\mathcal{C}^{\bar{x}}\geq 1$ in $\tau'(S, g(\id(p_i)))$}{
    Set the flag $\mathcal{F}$ in $\tau'(S, g(\id(p_i)))$ \;
    Update the flags $\mathcal{F}$ in the nodes along the path from $\tau'(S, g(\id(p_i)))$ to the root of $\tau'(S)$ \;
    Set $\mathcal{E}^{\bar{x}}=0$ in $\tau'(S, \id(p_i))$ \;
    \For{all $g'$ that is an $\epsilon$-neighbor of $g(\id(p_i))$ where the the local counter $\mathcal{C}^{\bar{x}} \geq 1$ in $\tau'(S, g')$ \label{line:bneighbors2}} {
      Decrease $\mathcal{E}^x$ of $\tau'(S, g')$ by 1\;
      \If{$\mathcal{E}^x$ in $\tau'(S, g')$ was decreased from 1 to 0} {
        Unset the flag $\mathcal{F}$ in $\tau'(S, g')$ \;
        Update the flags $\mathcal{F}$  in the nodes along the path from $\tau'(S, g')$ to the root of $\tau'(S)$ \;
      }
    }
  }
  \caption{Updating nodes for insertion for the bichromatic case.()}
  \label[procedure]{proc:update_insertion_ebcp}
\end{procedure}

\paragraph{Deletion.}
Given $(i, p_i,x)$, we perform the following steps to delete this tuple from our data structure.
\begin{enumerate}
  \item Compute the id, $\id(p_i)$, of the $\frac{\epsilon\xi}{\sqrt{d}}$-box which $p_i$ belongs to and denote it by $g(\id(p_i))$.
  \item Using $\id(p_i)$ as the key, we find the leaf node in $\tau'(S)$ that is corresponding to $g(\id(p_i))$.
  \item Remove $i$ from the skip list for color $x$, and decrease the local counter $\mathcal{C}^{x}$ in $\tau'(S, g(\id(p_i)))$ by 1.
  \item Use \cref{proc:update_deletion} to update the external counters $\mathcal{E}^x$ and $\mathcal{E}^{\bar{x}}$ (here $\bar{x}$ denote the other color than $x$) and flags $\mathcal{F}$ in $\tau'(S, g(\id(p_i)))$ 
  as well as in leaf nodes corresponding to the $\epsilon$-neighbors of $g(\id(p_i))$.
  \item If both local counters $\mathcal{C}^x, \mathcal{C}^{\bar{x}}$ in this leaf node are 0, remove $\tau'(S, g(\id(p_i)))$ from $\tau'(S)$, and update the bitmap $\mathcal{B}$ in $\tau'(S)$ that keeps track of all free memory cells. 
  \item Remove $(i, p_i,x)$ from the hash table.
\end{enumerate}

Note that the first step can be done in $O(d)$ time. The second step can be done in $O(\log n)$ time. \cref{proc:update_deletion_ebcp} has the same time complexity with \cref{proc:update_insertion_ebcp}. Hence, the cost for the deletion procedure is the same with that for insertion.

\begin{procedure}
  \SetKwInOut{Input}{input}\SetKwInOut{Output}{output}
  \Input{$(i, p_i, x)$ from $A$, the leaf node in $\tau'(S)$ corresponding to $g(\id(p_i))$, which we denote as $\tau'(S, g(\id(p_i)))$.}
  Let $\bar{x}\in \{A,B\}$ and $\bar{x}\neq x$\;
  \uIf{$\mathcal{C}^x$ and $\mathcal{C}^{\bar{x}}$ in $\tau'(S, \id(p_i))$  $=0$}{
    Unset the flag $\mathcal{F}$ in $\tau'(S, g(\id(p_i)))$ \;
    Update the flags $\mathcal{F}$ in the nodes along the path from $\tau'(S, \id(p_i))$ to the root of $\tau'(S)$ \;
    Set $\mathcal{E}^{x} = 0$ and $\mathcal{E}^{\bar{x}} = 0$ in $\tau'(S, \id(p_i))$ \;
    \For{all $g'$ that is an $\epsilon$-neighbor (see \cref{defn:kneighbor}) of $g(\id(p_i))$ where the local counter $\mathcal{C}^{\bar{x}} \geq 1$ in $\tau'(S, g')$ \label{line:bneighbors3}} {
      Decrease $\mathcal{E}^x$ of $\tau'(S, g')$ by 1\;
      \If{$\mathcal{E}^x$ in $\tau'(S, g')$ was decreased from 1 to 0} {
        Unset the flag $\mathcal{F}$ in $\tau'(S, g')$ \;
        Update the flags $\mathcal{F}$  in the nodes along the path from $\tau'(S, g')$ to the root of $\tau'(S)$ \;
      }
    }
  }
  \ElseIf{$\mathcal{C}^x=0$ and $\mathcal{C}^{\bar{x}}\geq 1$}{
    \For{all $g'$ that is an $\epsilon$-neighbor of $g(\id(p_i))$ where the local counter $\mathcal{C}^x \geq 1$ in $\tau'(S, g')$ \label{line:bneighbors4}} {
      Increase $\mathcal{E}^{\bar{x}}$ of $\tau'(S, g')$ by 1\;
      Increase $\mathcal{E}^x$ of $\tau'(S, g(\id(p_i)))$ by 1\;
      \If{$\mathcal{E}^{\bar{x}}$ in $\tau'(S, g')$ was increased from 0 to 1} {
        Set the flag $\mathcal{F}$ in $\tau'(S, g')$ \;
        Update the flags $\mathcal{F}$ in the nodes along the path from $\tau'(S, g')$ to the root of $\tau'(S)$ \;
      }
    }
    \If{$\mathcal{E}^x = 0$ in $\tau'(S, g(\id(p_i)))$} {
      Unset the flag $\mathcal{F}$ in $\tau'(S, g(\id(p_i)))$ \;
      Update the flags $\mathcal{F}$ in the nodes along the path from $\tau'(S, \id(p_i))$ to the root of $\tau'(S)$ \;
    }
  }
  \caption{Updating nodes for deletion for the bichromatic case.()}
  \label[procedure]{proc:update_deletion_ebcp}
\end{procedure}

\paragraph{Checking for $(1+\xi)\epsilon$-close pairs.} To check the existence of an $(1+\xi)\epsilon$-close pair, we just read the flag in the root of the radix tree. If the flag is set, there is at most one $\epsilon$-close pair in $S$, and no such pairs otherwise. This operation takes $O(1)$ time. 

\paragraph{Finding a $(1+\xi)\epsilon$-close pair.} We just read the flag in the root of the radix tree and then go to a leaf which flag is $1$. Check the local counters of the node. If both local counters are at least $1$, output the first elements in skip lists for $A$ and the first element in the skip list for $B$. Otherwise, check the external counters. Suppose $\mathcal{E}^A$ is non-zero. Then we find the $\epsilon$-neighbor of the current node whose $\mathcal{C}^B>0$ and output the first point in the skip list of $A$ of the current node and the first element in the skip list of $B$ of the $\epsilon$-neighbor. 

We have the following result.
\begin{theorem}
  \label{thm:ebcp}
  For any fixed dimension and fixed $\xi$, there exists a quantum algorithm that, with high probability, can solve $\xieBCP$ in time $O(n^{2/3}(d+\log^4(n+L) + d(\frac{4\sqrt{d}}{\xi}+1)^d))$.
\end{theorem}

\begin{proof}
  The proof closely follows the analysis for \cref{thm:ecp1}, and the correctness of the data structure and the time complexity of its operations follow from the discussion in \cref{sec:one-shot}. Note that our algorithm will output a pair which belong to the same $\frac{\xi\epsilon}{2\sqrt{d}}$-box or two of them that are $\epsilon$-neighbors. Based on \cref{lemma:represent}, two points which corresponding hyercubes are $\epsilon$-neighbors have distance at most $(1+\xi)\epsilon$. Therefore, our algorithm could output a pair of points which distance is at most $(1+\xi)\epsilon$. Another difference is that here we need to search at most $(4\sqrt{d}/\xi+1)^d$ neighbors during insertions and deletions. As a result, $\mathsf{U} = O(d+\log^4(n+L) + d(4\sqrt{d}/\xi+1)^d)$, and $\mathsf{S}=O(n^{2/3}(d+\log^4(n+L) + d(4\sqrt{d}/\xi+1)^d)$. Again, $\mathsf{C} = O(1)$, $\delta \geq 1/n^{2/3}$, and $\varepsilon \geq 1/n^{2/3}$. Therefore, by \cref{lem:framework}, the total cost is $O(\mathsf{S} + \frac{1}{\sqrt{\varepsilon}}(\frac{1}{\sqrt{\delta}}\mathsf{U} + \mathsf{C})) = O(n^{2/3}(d+\log^4(n+L) + (4\sqrt{d}/\xi+1)^d))$.
\end{proof}

By \cref{lemma:bcp-ebcp} and the above \cref{thm:ebcp}, we have the following theorem:
\begin{theorem}
  \label{thm:bcp-up}
  For an fixed dimension and fixed $\xi$, there exists a quantum algorithm that, with high probability, can solve $\xiBCP$ in time $\tilde{O}(n^{2/3})$.
\end{theorem}

\subsection{Quantum algorithm for solving $\BCP$ exactly}
\label{sec:bcp-exact}
In this subsection, we present a quantum algorithm for solving $\BCP$ exactly. The main idea of this algorithm is to partition $A$ into smaller subsets. Then we build data structures which support nearest-neighbor search on all subsets in superposition. We use the quantum minimum finding algorithm to find the smallest distances from $B$ to each subset, among which we use the quantum minimum finding algorithm again to find the smallest distance.   

Unlike the data structure for solving $\CP$, the data structure for $\BCP$ does not have to be uniquely represented, as no insertion and deletion are performed in this algorithm. The data structure can have expected running time instead of the worst-case running time. The total worst-case running time can be bounded by standard techniques. The nearest-neighbor search data structure we use is from~\cite{Clarkson88}, and is reformulated in the following lemma.

\begin{lemma}[\cite{Clarkson88}]
  \label{lemma:ds-bcp-nns}
  For any fixed dimension, there exists a data structure for $n$ points in $\R^d$ that can be built in expected time complexity $O(n^{\ceil{d/2} + \delta})$ for arbitrarily small $\delta$ and the nearest-neighbor search can be performed in worst-case time complexity $O(\log n)$.
\end{lemma}

This data structure is based on the Voronoi diagram and its triangulation in higher dimensions. 
Using this data structure, we have a quantum algorithm for solving $\BCP$ exactly, which yields the following theorem.

\begin{theorem}
  \label{thm:bcp-exact}
  There exists a quantum algorithm that, with high probability, solves $\BCP$ for dimension $d$ with time complexity $\tilde{O}\left(n^{1-\frac{1}{2d}+\delta}\right)$ for arbitrarily small $\delta$.
\end{theorem}

\begin{proof}
  We first partition $A$ into $\ceil{n/r}$ subsets $S_1, \dots, S_{\ceil{n/r}}$, where $|S_i| = r$ for $i \in \bigl[\ceil{n/r}\bigr]$. (The value of $r$ will be determined later). For all $i \in \bigl[\ceil{n/r}\bigr]$, we can find a closest pair between $S_i$ and $B$ as follows. First, a data structure as in \cref{lemma:ds-bcp-nns} for $S_i$ is built in expected time $O\bigl(r^{\ceil{d/2} + \delta}\bigr)$, which supports nearest-neighbor search in time $O(\log n)$. Then, we use the quantum minimum finding subroutine (\cref{thm:mim}) which uses the distance reported by the nearest-neighbor search as the oracle. The closest pair between $S_i$ and $B$ can be found in time complexity $\tilde{O}(\sqrt{n})$. Note that the time complexity for building the data structure is not bounded for the worst case. However, using Markov's inequality, we know that with high probability, say, at least $9/10$, the time complexity is bounded by $O\bigl(r^{\ceil{d/2} + \delta}\bigr)$. Hence, fixing a constant $c \geq 10$, and stop the data structure construction after $c \cdot r^{\ceil{d/2} + \delta}$ steps. With at most $1/10$ probability, the construction will fail and this event can be detected by checking the solution returned by the quantum minimum finding subroutine. We run $O(\log n)$ instances of above procedure in parallel and use take the quantum minimum of all the $O(\log n)$ results. The probability that all these instances fail is at most $(1/10)^{O(\log n)} = O(1/n)$. We refer to the above procedure as the ``inner search'', and its time complexity is $O\bigl(r^{\ceil{d/2} + \delta} + \sqrt{n}\bigr)$.
  
  Next, we use the distance of the output of the inner search as the oracle and perform another quantum minimum finding subroutine for $i \in \bigl[\ceil{n/r}\bigr]$. We refer to this procedure as the ``outer search''. The probability that the closest pair between $A$ and $B$ lies in $S_i$ and $B$ is $r/n$. As a result, the number of the oracle queries for the quantum minimum finding subroutine is $\tilde{O}(\sqrt{n/r})$. The time complexity for each query is $O\bigl(r^{\ceil{d/2} + \delta} + \sqrt{n}\bigr)$. Therefore, the total time complexity is $\wt{O}\bigl((r^{\ceil{d/2} + \delta} + \sqrt{n})\cdot \sqrt{n/r}\bigr)$. A simple calculation shows that this achieves the minimum (ignoring the $\delta$ term in the exponent) when $r = n^{1/d}/(d-1)^{2/d}$, which yields the total time complexity
  \begin{align}
    \tilde{O}\left(n^{1-\frac{1}{2d} + \delta}\right).
  \end{align}
  The failure probability for each query is at most $O(1/n)$. Therefore, the total failure probability is at most $O(\sqrt{n/r}/n)=O(n^{-(1/2-1/2d)})$ for $d>1$, which can be smaller than any constant.
\end{proof}

\subsection{Quantum lower bound for $\BCP$ in constant dimensions}
Now, we give a lower bound for $\xiBCP$, which trivially holds for $\BCP$.
\begin{theorem}
\label{thm:bcp-lo}
The quantum query complexity for solving $\BCP$ is $\Omega(n^{2/3})$. Furthermore, the quantum query complexity for solving $\xiBCP$ with an arbitrary $\xi$ is also $\Omega(n^{2/3})$.
\end{theorem}
\begin{proof}
Recall that we have shown in \cref{sec:cp-lo} that \textsf{ED} reduces to \textsf{CP} by viewing \textsf{ED} as one-dimensional $\CP$ with the minimum distance $0$. It is not hard to see that \textsf{ED} also reduces to approximate $\CP$ with multiplicative error $1+\xi$ since $0$ times $1+\xi$ is still $0$. For simplicity, we denote  approximate $\CP$ with multiplicative error $1+\xi$ as $\xiCP$. Given a set $S$ as a $\xiCP$ instance, we choose $A,B\subset S$ uniformly at random such that $A= S\setminus B$ and $|A|=|B|$. Then, with $1/2$ probability, a closest pair in $S$ has one point in $A$ and another in $B$. Therefore, if $(a,b)$ be a valid solution for $\xiBCP$ on $(A,B)$,
$(a,b)$ is also a a valid solution for $\xiCP$ on $S$ with probability $1/2$. 

It is obvious that following the same proof, $\CP$ reduces to $\BCP$. Hence, the quantum query complexity for $\BCP$ and $\xiBCP$ are both $\Omega(n^{2/3})$. This completes the proof.
\end{proof}

\section{Orthogonal vectors in constant dimensions}
\begin{theorem}\label{thm:ov_constant}
The time complexity of $\OV_{n,d}$ (\cref{def:OV}) in quantum query model is $\Theta(\sqrt{n})$ when the dimension $d$ is constant .
\end{theorem}

\begin{proof}

We show  lower and upper bounds for $\OV_{n,d}$:

\paragraph*{Lower bound} We reduce the search problem to an instance of 2-dimensional $\OV$.  Let all vectors in $A$ be $(0,1)$. We map an element of the search instance with value $0$ as a vector in $B$ with value $(0,1)$ in $\OV_{n,2}$, and $1$ as $(1,0)$. An orthogonal pair must contain the vector in $B$ with value $(1,0)$ in this construction. Therefore,  if we find an orthogonal pair, we find the corresponding marked (value $1$) element in the search instance.  The $\Omega(\sqrt{n})$ lower bound of Grover's search algorithm gives an $\Omega(\sqrt{n})$ lower bound to $\OV_{n,d}$.

\paragraph*{Upper bound} The vectors only have $2^d$ possible values, $ \zo^d$, in the $d$-dimensional \textsf{OV}. For a particular value $v\in \zo^d$, we can use Grover search to check whether there exist vector $a \in A$  such that $a=v$ in time $O(\sqrt{n})$, and similarly for vectors in $B$. Therefore we can, for all $v\in \zo^d$,  check whether there exist $a\in A$ such that $ a=v$ and $b\in B$ such that $ b=v$ in $O(2^{d+1} \sqrt  n)$ time, recording the results as two $2^{d}$ bit strings $S_A$ and $S_B$. Then we check all $2^{2d}$ pairs of values $(v,w)$ whether $\ip{v}{w}=0$ , $S_A(v)=1$, and $S_B(w)=1$.  When we found such a pair $(v,w)$, we use Grover's search algorithm again to output a corresponding pair of vectors. The total running time is $O(2^{d+1} \sqrt n +2^{2d}+2\sqrt n)=\wt{O}(\sqrt n)$.
\end{proof}

\section{Acknowledgements}
We would like to thank Lijie Chen and Pasin Manurangsi for helpful discussion. We would like to thank anonymous reviewers for their valuable suggestions on this paper.

\bibliographystyle{alpha}
\bibliography{ref}

\appendix
\end{document}